\documentclass[10pt,journal]{IEEEtran}

\usepackage[utf8]{inputenc} 
\usepackage[T1]{fontenc}
\usepackage{url}
\usepackage{ifthen}
\usepackage{cite}
\usepackage[cmex10]{amsmath} 
\usepackage{graphicx,colordvi,psfrag}
\usepackage{amsmath,amssymb}
\usepackage{epstopdf}
\usepackage[caption=false]{subfig}
\usepackage{epsfig,cite}
\usepackage{calc, pgf, xcolor}
\usepackage{nicefrac}
\usepackage{enumerate}
\usepackage{bm}
\usepackage{fixltx2e}
\usepackage{pstricks}
\usepackage{url}
\usepackage{setspace}
\usepackage{dsfont}

\newtheorem{remark}{Remark}

\newtheorem{lemma}{Lemma}
\newtheorem{proposition}{Proposition}
\newenvironment{proof}[1][Proof]{\noindent\textbf{#1.} }{\ \rule{0.5em}{0.5em}}

\newcommand{\bY}{\mathbf{Y}}
\newcommand{\bW}{\mathbf{W}}

\newcommand{\bX}{\mathbf{X}}

\newcommand{\bZ}{\mathbf{Z}}
\newcommand{\bh}{\mathbf{h}}
\newcommand{\bH}{\mathbf{H}}
\newcommand{\bA}{\mathbf{A}}
\newcommand{\ba}{\mathbf{a}}

\newcommand{\bb}{\mathbf{b}}

\newcommand{\bV}{\mathbf{V}}

\newcommand{\bE}{\mathbf{E}}

\newcommand{\bI}{\mathbf{I}}

\newcommand{\bL}{\mathbf{L}}
\newcommand{\bg}{\mathbf{g}}

\newcommand{\bs}{\mathbf{s}}

\newcommand{\ZZ}{\mathbb{Z}}

\newcommand{\RR}{\mathbb{R}}

\newcommand{\Unif}{\mathop{\mathrm{Unif}}}
\newcommand{\OL}{\mathcal{E}_{\overline{\mathrm{OL}}}}
\newcommand{\OLn}{\mathcal{E}_{\overline{\mathrm{OL}}_n}}

\DeclareMathOperator*{\argmin}{\arg\!\min}
\DeclareMathOperator*{\argmax}{\arg\!\max}

\newrgbcolor{BoxRed}{1 .5 .5}
\newrgbcolor{BoxBlue}{.9 .9 1}
\newrgbcolor{lightgrey}{.7 .7 0.7}

\begin{document}

\title{A Modulo-Based Architecture for Analog-to-Digital Conversion}
\author{Or Ordentlich, Gizem Tabak, Pavan Kumar Hanumolu, Andrew C. Singer and Gregory W. Wornell
\thanks{}
\thanks{O. Ordentlich is with the Hebrew University of Jerusalem, Israel (email: or.ordentlich@mail.huji.ac.il). G. Tabak, P. K. Hanumolu and A. C. Singer are with the University of Illinois, Urbana-Champaign, USA (emails: \{tabak2,hanumolu,acsinger\}@illinois.edu). G. W. Wornell is with the Massachusetts Institute of Technology, MA, USA (email: gww@mit.edu)}
\thanks{}}

\maketitle

\begin{abstract}
Systems that capture and process analog signals must first acquire them through an analog-to-digital converter. While subsequent digital processing can remove statistical correlations present in the acquired data, the dynamic range of the converter is typically scaled to match that of the input analog signal. 
The present paper develops an approach for analog-to-digital conversion that aims at minimizing the number of bits per sample at the output of the converter. This is attained by reducing the dynamic range of the analog signal by performing a modulo operation on its amplitude, and then quantizing the result. While the converter itself is universal and agnostic of the statistics of the signal, the decoder operation on the output of the quantizer can exploit the statistical structure in order to unwrap the modulo folding. The performance of this method is shown to approach information theoretical limits, as captured by the rate-distortion function, in various settings. An architecture for modulo analog-to-digital conversion via ring oscillators is suggested, and its merits are numerically demonstrated.
\end{abstract}

\section{Introduction}
\label{sec:intro}

Analog-to-digital converters (ADCs) are an essential component in any device that manipulates analog signals in a digital manner. While digital systems have benefited tremendously from scaling, their analog counterparts have become increasingly challenging. Consequently, it is often the case that the ADC constitutes the main bottleneck in a system, both in terms of power consumption and real estate, and in terms of the quality of the system's output. Developing more efficient ADCs is therefore of great interest~\cite{Walden99,lrrb05}.

The quality of an ADC is measured via the tradeoff between various parameters such as power consumption, size, cost of manufacturing, and the distortion between the input signal and its digitally-based representation. For the sake of a unified, technology-independent, discussion, it is convenient to restrict the characterization of an ADC quality to three basic parameters: 1) The number of analog samples per second $F_S$; 2) The number of ``raw'' output bits $R$ the ADC produces per sample (before subsequent possible compression); 3) The mean squared error (MSE) distortion $D$ between the input signal and a reconstruction that is based on the output of the ADC.

While different applications may require different tradeoffs between $F_S$, $R$ and $D$, it is always desirable to design the ADC such that all three parameters are as small as possible. The focus of this work is on the quantization rate $R$. For a given sampling frequency $F_S$, and a given target distortion $D$, our goal is to design ADCs that use the smallest possible number of raw output bits per sample.

The problem of analog-to-digital conversion can be seen as an instance of the lossy source coding/lossy compression problem~\cite{berger71,jayantnoll,coverthomas}, as the output of an ADC is a binary sequence, which represents the analog source. A unique key feature of the analog-to-digital conversion problem is that the encoding of the source is carried out in the analog domain, while the decoding procedure is purely digital. Given the limitations of analog processing, it is therefore generally only practical to exploit the source structure at the decoder. Hence, the type of source coding schemes that are suitable for data conversion, are those that approach fundamental limits without requiring knowledge of the source structure at the encoder. In addition, latency and complexity constraints in data conversion, typically preclude the use of schemes other than those based on scalar quantization.

The input signal to an ADC is often known to have structure that could be exploited to reduce the overall bit rate of its representation, $R$.
In our analysis, it will be convenient to express this structure using a stochastic model for the input. Consequently, throughout the paper, we will model the input to the ADC as a stationary stochastic Gaussian process $X(t)$, whose power spectral density (PSD) encapsulates the assumed structure. More generally, we will sometimes also consider the problem of analog-to-digital conversion of a vector $\bX(t)=\{X_1(t),\ldots,X_K(t)\}$ of jointly stationary stochastic Gaussian processes, via $K$ parallel ADCs, the input to each one of them is one of the $K$ processes. 

Under such stochastic modeling, rate-distortion theory~\cite{berger71} provides the fundamental lower bound $F_s\cdot R>R_X(D)$ for any ADC (and corresponding decoder) that achieves distortion $D$, where $R_X(D)$ is the rate-distortion function of the process $X(t)$ in bits per second. In general, achieving the rate-distortion function of a source requires using sophisticated high-dimensional quantizers, whereas analog-to-digital conversion is invariably done via scalar uniform quantizers. Thus, achieving this lower bound with ADCs seems overly optimistic. Nevertheless, as we shall see, approaching the rate-distortion bound, up to some inevitable loss due to the one-dimensional nature of the quantization, is sometimes possible by a simple modification of the scalar uniform quantizer, namely, a \emph{modulo ADC}, followed by a digital decoder that efficiently exploits the source structure.

\begin{figure}[t]
\begin{center}
\psset{unit=0.6mm}
\begin{pspicture}(0,10)(130,57)

\rput(-5,30){$x$} \psline{->}(0,30)(15,30)
\psframe(15,15)(75,43.5)
\rput(45,48){$\bmod \Delta$}
\psline(20,21.5)(37,38.5)(37,21.5)(54,38.5)(54,21.5)(71,38.5)
\rput(70.5,30){$\ldots$}
\rput(19.5,30){$\ldots$}

\rput(75,0){
\psline{->}(0,30)(10,30)
\psframe(10,15)(35,43.5)
\rput(21,48){$Q(\cdot)$}
\rput(10,30){
\psline(0,-6.8)(7.4,-6.8)(7.4,-3.4)(10.8,-3.4)(10.8,0)(14.2,0)(14.2,3.4)(17.6,3.4)(17.6,6.8)(25,6.8)
}
}

\psline{->}(110,30)(125,30)

\rput(70,62){Modulo ADC}
\psframe[linestyle=dashed](7,10)(118,56)

\end{pspicture}
\end{center}
\caption{A schematic illustration of the modulo ADC.}
\label{fig:modADC}
\end{figure}
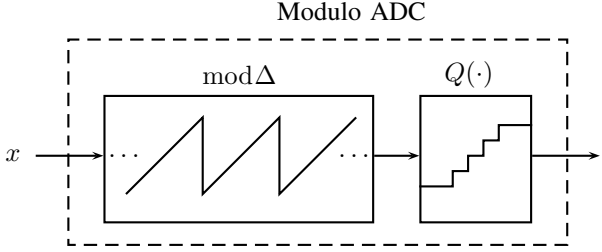

Instead of sampling and quantizing the process $X(t)$, a modulo ADC samples and quantizes the process $[X(t)]\bmod \Delta$, where the modulo size $\Delta$ is a design parameter. See Figure~\ref{fig:modADC}. Equivalently, a modulo ADC can be thought of as a standard uniform scalar ADC with step-size $\delta$ and an arbitrarily large dynamic range/support, but that outputs only the $R$ least significant bits in the description of each sample, where $2^R=\tfrac{\Delta}{\delta}$. The benefit of applying the modulo operation on $X(t)$ is in reducing its dynamic range/support, which in turn enables a reduction of the number of bits per sample produced by the ADC, without increasing the quantizer's step-size. 
This operation, which corresponds to disregarding coarse information about $X(t)$, will
otherwise substantially degrade the source reconstruction. However, by properly accounting for the modulo operation and
appropriately choosing its parameter $\Delta$, we can unwrap the modulo operation with high probability using previous samples
of $X(t)$ and exploiting the (redundant) structure in the signal.

Following standard system design methodology, in the performance analysis of a modulo ADC, we distinguish between two events: 1) The no-overload event $\bar{\mathcal{E}}_{\mathrm{OL}}$ where the decoder was able to correctly unwrap the modulo operation. We require the MSE distortion, conditioned on this event, to be at most $D$; 2) The overload event $\mathcal{E}_{\mathrm{OL}}$ where the decoder fails in unwrapping the modulo operation. We require the probability of this event $\Pr(\mathcal{E}_{\mathrm{OL}})$ to be small, but do not concern ourselves with the MSE distortion conditioned on the occurrence of this event.

\subsection{Our Contributions}

This work further develops the modulo ADC framework in three complementary directions, as specified below.

\subsubsection{Oversampled Modulo ADC}
We show that a modulo ADC can be used as an alternative to $\Sigma\Delta$ converters. A $\Sigma\Delta$ converter is based on oversampling the input process $X(t)$, i.e., sampling above the Nyquist rate, in conjunction with noise-shaping, which pushes much of the energy of the quantization noise to high frequencies, where there is no signal content. See Figure~\ref{fig:SigDel}. The noise shaping operation requires incorporating an elaborate mixed signal feedback circuit. In particular, the circuit first generates the quantization noise, which necessitates using not only an ADC, but also an accurately-matched digital-to-analog converter (DAC), and then applies an analog filter. The analog nature of the signal processing makes it challenging to use filters of high-orders, which in turn limits performance. 

We develop an alternative architecture (Section~\ref{sec:oversampled}) that shifts much of the complexity to the decoder, whereas the ``encoder'' is simply a modulo ADC. See Figure~\ref{fig:temporal}. The parameter $\Delta$ in the modulo ADC, as well as the coefficients of the prediction filter in Figure~\ref{fig:temporal}, depend only on the bandwidth $B$ of the input process $X(t)$ and on its variance $\sigma^2$, and not on the other details of its PSD. Similarly, the MSE distortion between the input process and its reconstruction, depends only on $B$ and $\sigma^2$. Thus, the developed architecture is as agnostic as $\Sigma\Delta$ converters to the statistics of the input process. Furthermore, for a flat-spectrum process, the distortion is within a small gap, due to one-dimensionality of the encoder, from the information theoretic limit.

\begin{figure}[t]
\begin{center}
\psset{unit=0.6mm}
\begin{pspicture}(0,0)(250,40)

\rput(0,30){$X_n$}\psline{->}(5,30)(13,30)
\rput(-17,0){
\pscircle(35,30){5}\rput(35,30){$\Sigma$}\rput(28,27){$-$}
\psline{->}(40,30)(82,30)
\psframe(82,25)(97,35)\rput(90,30){ADC}
\psline{->}(97,30)(122,30)
\psframe(122,25)(137,35)\rput(129,30){LPF}
\psline{->}(137,30)(150,30)\rput(155,30){$\hat{X}_n$}

\psline{->}(105,30)(105,10)(97,10)
\psframe(82,15)(97,5)\rput(90,10){DAC}
\psline{->}(82,10)(75,10)

\pscircle(70,10){5}\rput(70,10){$\Sigma$}\rput(77,7){$-$}
\psline{->}(70,30)(70,15)

\psline{->}(65,10)(55,10)
\psframe(40,5)(55,15)\rput(48,10){Filter}
\psline{->}(40,10)(35,10)(35,25)
\psframe[linestyle=dashed](21,0)(108,40)\rput(28,43){Encoder}
\psframe[linestyle=dashed](115,20)(147,40)\rput(125,43){Decoder}
}

\end{pspicture}
\end{center}
\caption{Schematic architecture for oversampled $\Sigma\Delta$ converter. $\{X_n\}$ is obtained by sampling the process $X(t)$.}
\label{fig:SigDel}
\end{figure}
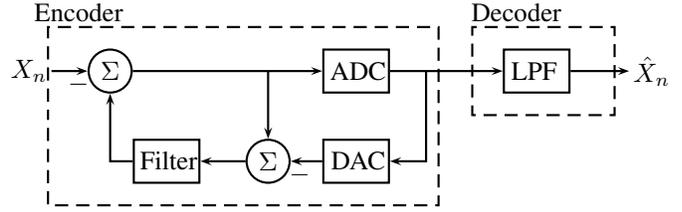 
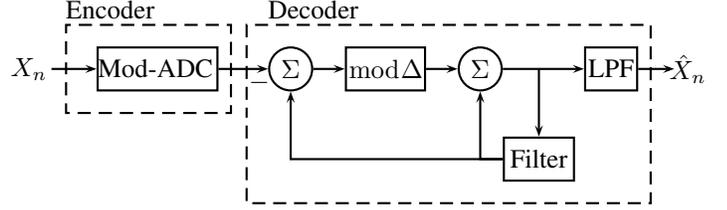
\begin{figure}[t]
\begin{center}
\psset{unit=0.6mm}
\begin{pspicture}(0,0)(250,40)

\rput(0,30){$X_n$}\psline{->}(5,30)(15,30)
\psframe(15,25)(42,35)\rput(28,30){Mod-ADC}
\psline{->}(42,30)(53,30)

\psframe[linestyle=dashed](8,20)(45,40)\rput(18,43){Encoder}

\rput(-27,0){
\pscircle(85,30){5}\rput(85,30){$\Sigma$}\rput(78,27){$-$}
\psline{->}(90,30)(97,30)\psframe(97,25)(115,35)\rput(106,30){$\bmod\Delta$}
\psline{->}(115,30)(122,30)\pscircle(127,30){5}\rput(127,30){$\Sigma$}
\psline{->}(132,30)(150,30)
\psframe(150,25)(162,35)\rput(156,30){LPF}
\psline{->}(162,30)(170,30)
\rput(173,30){$\hat{X}_n$}
\psline{->}(140,30)(140,15)\psframe(132,5)(148,15)\rput(140,10){Filter}
\psline{->}(132,10)(85,10)(85,25)
\psline{->}(132,10)(127,10)(127,25)
}

\psframe[linestyle=dashed](48,0)(138,40)\rput(63,43){Decoder}

\end{pspicture}
\end{center}
\caption{Schematic architecture for oversampled modulo ADC. The same architecture, without the low-pass filter (LPF) is also suitable for modulo ADC for a general stationary process. $\{X_n\}$ is obtained by sampling the process $X(t)$.}
\label{fig:temporal}
\end{figure}

\subsubsection{A Phase-Domain Implementation of Modulo ADC via Ring Oscillators}

We develop a modulo ADC implementation that performs the modulo reduction inherently as part of the analog signal acquisition process. As the phase of a periodic waveform is always measured modulo $2\pi$, a natural class of candidates are ADCs that first convert the input voltage into phase, and then quantize that phase. A notable representative within this class, which has been extensively studied in the literature\cite{holt97,sp08}, is the \emph{ring oscillator ADC}.

\begin{figure}[t]
	\includegraphics[scale=0.2]{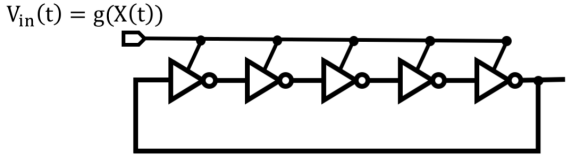}
	\caption{A schematic illustration of a ring oscillator with $N=5$ inverters. The states of all $N$ inverter are measured every $T_S$ seconds.}
	\label{fig:ringosc}
\end{figure}

\begin{figure}[t]
	\includegraphics[width=0.5\textwidth]{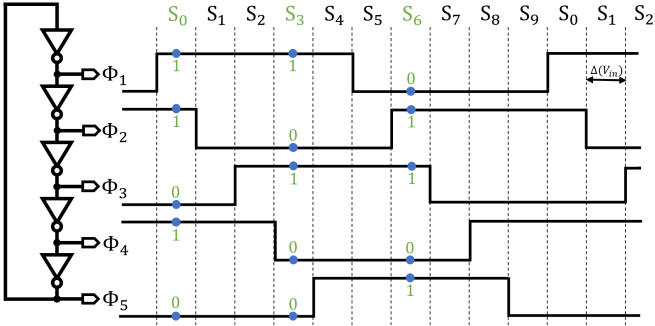}
	\caption{An example of the evolution of the states of the inverters in a ring oscillator.}
	\label{fig:inverterstates}
\end{figure}

Consider a closed-loop cascade of $N$ inverters, where $N$ is an odd number, all controlled with the same voltage $V_{dd}=V_{\text{in}}$, see Figure~\ref{fig:ringosc}. This circuit, which will be described in detail in Section~\ref{sec:ringosc}, oscillates between $2N$ states, corresponding to the values (`low' or `high', represented by `$0$' or `$1$') of each of the $N$ inverters. See Figure~\ref{fig:inverterstates}. The oscillation frequency is controlled by $V_{dd}$. Due to the oscillating nature of the circuit, if we sample its state every $T_S$ seconds, we cannot tell how many ``state changes" occurred between two consecutive samples, but we are able to determine this number modulo $2N$. Thus, by setting  $V_{dd}$ to $V_{dd}(t)=g(X(t))$, where $X(t)$ is the analog signal to be converted to a digital one and $g(\cdot)$ is a function to be specified, we obtain a modulo ADC. The input-output relation of this modulo ADC is characterized in Section~\ref{sec:ringosc}, and depends on the response time of the inverters to change in their input, as a function of $V_{dd}$.

In practice, the modulo operation realized in this way deviates from the ideal characteristic of Figure~\ref{fig:modADC} in a variety of ways. Accordingly, we perform several numerical experiments to evaluate and optimize the performance of an oversampled ring oscillator modulo ADC, and compare it to the performance of an ideal modulo ADC as well as to a $\Sigma\Delta$ converter. The results demonstrate that despite the non-idealities in the ring oscillator implementation, in some regimes, this architecture holds substantial potential for improvement over existing ADCs.

\subsubsection{Modulo ADCs for Jointly Stationary Processes}

In many applications the number of sensors/antennas observing a particular process is greater than the number of degrees-of-freedom (per time unit) governing its behavior. Thus, there is a redundancy at the receiver that can be exploited.  However, as this redundancy can be spread across time and space, traditional ADC architectures, as well as the modulo ADC architectures described in Section~\ref{subsec:timecor} and~\ref{subsec:spacecor}, are insufficient. In this part of the paper, we show how to address this problem via a natural extension of the modulo ADC framework.

As an example we will consider the problem of wireless communication.
It is by now well established that using receivers, as well as transmitters, with multiple antennas, dramatically increases the achievable communication rates over wireless channels~\cite{telatar99,tseviswanath}. However, adding antennas comes with the price of requiring multiple expensive and power hungry RF chains. For traditional ADC architectures, power and cost scale linearly with the number of receive antennas, which motivates an alternative solution.

It is often the case, that the signals observed by the different receive antennas are highly correlated, in time and in space. As an illustrative example, consider the case where the transmitter has one antenna, whereas the receiver has $K>1$ antennas. We can model the signal observed at each of the antennas, after sampling, as
\begin{align}
Y^k_n=h^k_n*X_n+Z^k_n, \ k=1,\ldots,K, \ n=1,\ldots,N, \label{eq:timeMIMO}
\end{align}
where $\{X_n\}$ is the process emitted by the transmitter, $\{h^k_n\}$ is the $k$th channel impulse response, and $\{Z^k_n\}$ are independent additive white Gaussian noise (AWGN) processes. 

Since all $K$ output processes $\{Y^1_n\},\ldots,\{Y^K_n\}$ in~\eqref{eq:timeMIMO} are noisy and filtered versions of of the same input process, they will typically be highly correlated. However, this correlation may be spread in time (the $n$-axis) and in space (the $k$-axis). As an extreme example, assume $\{X_n\}$ is an iid process, and the filters simply incur different delays, i.e., $h^k_{n}=\delta_{n-k}$ for $k=1,\ldots,K$. While each individual process $\{Y_k^n\}$ is white, and each vector $(Y_n^1,\ldots,Y_n^K)$, $n=1,\ldots,N$ has a scaled identity covariance matrix, the vector process $\{\{Y^1_n\},\ldots,\{Y^K_n\}\}$ is highly correlated. One must therefore jointly process the time and the spatial dimensions in order to exploit this correlation.

This phenomenon, where the signals observed by the different ADCs are highly correlated, is not unique to the wireless communication setup, and appears in many other applications, e.g., multi-array radar. It is, however, taken to the extreme in \emph{massive MIMO}~\cite{letm14}, where the number of antennas at the base station is of the order of tens or even hundreds, while the number of users it supports may be substantially fewer. 

In Section~\ref{subsec:spacetimecor} we develop an architecture that uses modulo ADCs, one for each receive antenna, in order to exploit the space-time correlation of the processes. We develop a low-complexity decoding algorithm for unwrapping the modulo operations. This algorithm combines the idea of performing prediction in time, of the quantized vector process from its past, with that of integer-forcing source decoding~\cite{oe13b}, which is used for exploiting spatial correlations in the prediction error vector. See Figure~\ref{fig:spatiotemporal}. In the limit of small $D$, the loss of the developed analog-to-digital conversion scheme with respect to the information theoretic lower bound on $D$, is shown to reduce to that of the integer-forcing source decoder.

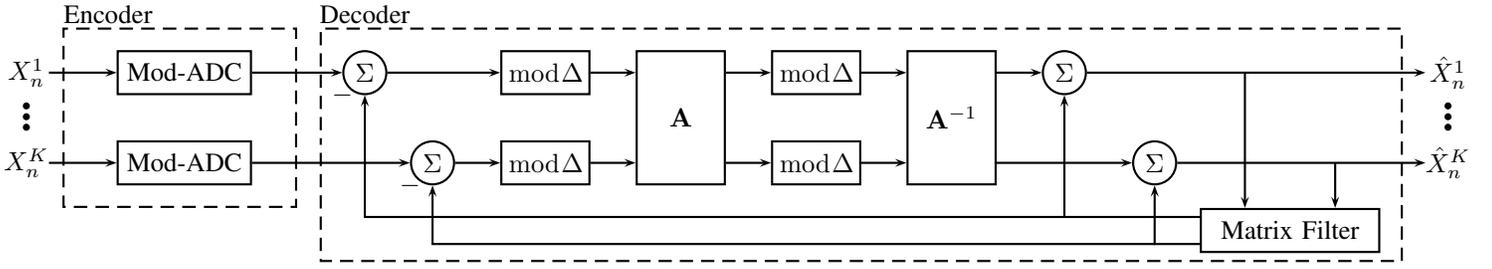
\begin{figure*}[t]
\begin{center}
\psset{unit=0.6mm}
\begin{pspicture}(0,-13)(310,45)

\rput(-5,30){$X^1_{n}$}\psline{->}(0,30)(15,30)
\psframe(15,25)(45,35)\rput(30,30){Mod-ADC}
\psline{->}(45,30)(65,30)
\pscircle(70,30){5}\rput(70,30){$\Sigma$}\rput(65,25){$-$}
\psline{->}(75,30)(100,30)
\psframe(100,25)(120,35)\rput(110,30){$\bmod\Delta$}\psline{->}(120,30)(130,30)

	\rput(-5,22){\Huge$\vdots$}

\rput(0,-20){
	\rput(-5,30){$X^K_{n}$}\psline{->}(0,30)(15,30)
		\psframe(15,25)(45,35)\rput(30,30){Mod-ADC}
		\psline{->}(45,30)(80,30)
		\pscircle(85,30){5}\rput(85,30){$\Sigma$}\rput(80,25){$-$}
		\psline{->}(90,30)(100,30)
		\psframe(100,25)(120,35)\rput(110,30){$\bmod\Delta$}\psline{->}(120,30)(130,30)
	}

\rput(50,0){
\psframe(80,5)(100,35)\rput(90,20){$\bA$}
\psline{->}(100,30)(110,30)\psframe(110,25)(130,35)\rput(120,30){$\bmod\Delta$}\psline{->}(130,30)(140,30)
\psline{->}(100,10)(110,10)\psframe(110,5)(130,15)\rput(120,10){$\bmod\Delta$}\psline{->}(130,10)(140,10)
\psframe(140,5)(160,35)\rput(150,20){$\bA^{-1}$}
\psline{->}(160,30)(170,30)\pscircle(175,30){5}\rput(175,30){$\Sigma$}\psline{->}(180,30)(255,30)\rput(260,30){$\hat{X}_n^1$}
	\rput(260,22){\Huge$\vdots$}
\psline{->}(160,10)(190,10)\pscircle(195,10){5}\rput(195,10){$\Sigma$}\psline{->}(200,10)(255,10)\rput(260,10){$\hat{X}_n^K$}
\psline{->}(215,30)(215,0)
\psline{->}(235,10)(235,0)
\psframe(205,-10)(245,0)\rput(225,-5){Matrix Filter}
\psline{->}(205,-2)(175,-2)(175,25)
\psline{->}(205,-8)(195,-8)(195,5)
\psline{->}(175,-2)(20,-2)(20,25)
\psline{->}(195,-8)(35,-8)(35,5)
}

\psframe[linestyle=dashed](3,0)(55,40)\rput(13,43){Encoder}

\psframe[linestyle=dashed](60,-12)(300,40)\rput(70,43){Decoder}

\end{pspicture}
\end{center}
\caption{Schematic architecture for Modulo ADCs for jointly stationary processes.}
\label{fig:spatiotemporal}
\end{figure*}

\subsection{Related Work}

The idea of using modulo ADCs/quantizers for exploiting temporal correlations within the input process $X(t)$ towards reducing the quantization rate $R$, dates back, at least, to~\cite{er79}, where a quantization scheme, called modulo-PCM, was introduced. A decoding scheme for unwrapping the modulo operation, based on maximum-likelihood sequence detection~\cite{forney72}, was further proposed in~\cite{er79}, and a heuristic analysis was performed, based on prediction of $X(t)$ from its past, which shows that modulo-PCM can approach the Shannon lower bound under the high-resolution assumptions. In Section~\ref{subsec:timecor}, we develop a more complete analysis of modulo quantization, the details of which are required for the application we discuss in Section~\ref{sec:oversampled}.

The architecture from Figure~\ref{fig:temporal} is based on using a prediction filter at the decoder, as a part of the modulo unwrapping process,  as was hinted at in~\cite{er79} (see also~\cite{ramamoorthy85}). In agreement with the literature on differential pulse-code modulation (DPCM) at the late 1970s (see e.g.~\cite{noll78}), the authors in~\cite{er79} proposed to design the prediction filter as the optimal one-step predictor of the unquantized process $\{X_n\}$ from its past. As shown in~\cite{zke08}, this design criterion is sub-optimal, and the ``correct'' design criterion is to take this filter as the one-step predictor of the \emph{quantized} process from its past. The difference between the two design criteria is significant for oversampled processes, which are the focus of Section~\ref{sec:oversampled}, whose PSD is zero at high frequencies, as in those frequencies the signal-to-distortion ratio is zero, no matter how small the quantization noise is. Our analysis in Section~\ref{sec:oversampled} reveals that designing the modulo size $\Delta$ and the prediction filter with respect to a quantized flat-spectrum input process, results in a \emph{universal} system. This means, that this system attains the same distortion $D$ for all input processes
that share the same support for the PSD and the same variance.

The use of modulo ADCs/quantizers was also studied by Boufounos in the context of quantization of oversampled signals~\cite{boufounos12} (see also~\cite{vb16}). In particular, it is shown in~\cite{boufounos12} that by randomly embedding a measurement vector in $\RR^K$ onto an $M\gg K$ dimensional subspace, and using a modulo ADC for quantizing each of the coordinates of the result, one can attain a distortion that decreases exponentially with the oversampling ratio, with high probability. In Section~\ref{sec:oversampled} we consider a similar setup, where an oversampled analog signal, with oversampling ratio $L>1$, i.e. $F_s$ is $L$ times greater than the Nyquist frequency, is digitized by a modulo ADC. In the language of~\cite{boufounos12}, this corresponds to embedding $\bX\in\RR^K$ to an $M=LK$ dimensional space by zero-padding followed by interpolation, which is indeed a linear operation. We show that for this particular ``embedding'' not only is the decay of MSE distortion exponential in the oversampling ratio, but the attained distortion is information-theoretically optimal, up to a constant loss, which is explicitly characterized, due to the scalar nature of the quantizer. Moreover, under this ``embedding'', a simple low-complexity decoding algorithm exists, whereas for the random projection case studied in~\cite{boufounos12}, no computationally efficient decoding algorithm was given. One advantage, on the other hand, of the approach from~\cite{boufounos12}, is that it is applicable to $1$-bit modulo ADCs, whereas the performance of the scheme from Section~\ref{sec:oversampled} typically becomes attractive starting from $R \gtrsim 2$ bits per sample.

Very recently, Bhandari \emph{et al.} have addressed the question of what is the minimal sampling rate that allows for exact recovery of a bandlimited finite-energy signal, from its modulo-reduced sampled version~\cite{bkr17} (see also~\cite{bkr18}). They have found that a sufficient condition for correct reconstruction is sampling above the Nyquist rate by a factor of  $2\pi e$, regardless of the size of the modulo interval.
The analysis in~\cite{bkr17} did not take quantization noise into account, which corresponds to $R=\infty$ and $D=0$ in our setup. 

The merits of a modulo ADC for distributed analog-to-digital conversion of signals correlated in space, but not in time, were demonstrated in~\cite{oe13b}. A low-complexity decoding algorithm, for unwrapping the modulo operation, was proposed and its performance was analyzed. It was demonstrated via numerical experiments that the performance is usually quite close to the information theoretic lower bounds (See also~\cite{de17}). In Section~\ref{subsec:spacecor}, we summarize the decoding scheme from~\cite{oe13b} and the corresponding performance analysis, as those will be needed in Section~\ref{subsec:spacetimecor}, where we develop a modulo ADC architecture for analog-to-digital conversion of jointly stationary processes. The decoding algorithm for this setup, as well as its performance analysis, is inspired by the ideas
and techniques from Sections~\ref{subsec:timecor} and~\ref{subsec:spacecor}.

In a broader sense, modulo quantization is closely related to Wyner-Ziv's source coding with side information setup and to its channel coding dual, which is the Gel'fand-Pinsker setup~\cite{zse02}. In the latter context, we further note that modulo quantization is widely used for communication over intersymbol interference channels~\cite{tomlinson71,harashima72}. Recently, Hong and Caire~\cite{hc13IT} considered modulo ADCs as potential candidates for the front end of receivers in a cloud radio access network (CRAN), employing compute-and-forward~\cite{ng11IT} based protocols.

Note that the although the concept of modulo ADC is reminiscent of \emph{folding ADCs}~\cite{foldingADC}, an important difference is that unlike the latter, the former does not keep track of the number of folds that occurred and, moreover, its functionality does not depend on this number, i.e., it does not saturate for large inputs. In unwrapping the modulo operation at the decoder, the missing information about number of folds is recovered, and we are able to attain the same $D$ with smaller rate.

Finally, another related line of work, is that of compressed sampling, see, e.g.,~\cite{venkataramani2000perfect,vetterli2002sampling,me10}, where the goal is to design universal and efficient ADCs  with a small sampling frequency $F_S$, under the assumption that the input signal occupies only a small portion of its total bandwidth, but the exact support is unknown.

\subsection{Organization}

The rest of the paper is organized as follows. In Section~\ref{sec:ideal} we formally define the modulo ADC and study its performance for stationary scalar input processes, and for random vectors (spatial correlation). Section~\ref{sec:oversampled} develops the use of oversampled modulo ADCs as a substitute for $\Sigma\Delta$ converters, and analyzes the tradeoffs this architecture achieves. In Section~\ref{sec:ringosc} we introduce an implementation of modulo ADCs via ring oscillators and establish the corresponding input-output mathematical model. Numerical experiments for evaluating the performance of ring oscillators based oversampled modulo ADCs are performed in Section~\ref{sec:numerical}. Section~\ref{subsec:spacetimecor} proposes to use parallel modulo ADCs for digitizing jointly stationary processes. The paper concludes in Section~\ref{sec:conc}.

\section{Preliminaries on Ideal Modulo ADC}
\label{sec:ideal}

Let $\Delta\in \RR^+$ be a positive number, and define the $\bmod \Delta$ operation as
\begin{align}
[x]\bmod\Delta\triangleq x-\Delta\left\lfloor \frac{x}{\Delta}\right\rfloor\in[0,\Delta),\nonumber
\end{align}
where the floor operation $\lfloor x\rfloor$ returns the largest integer smaller than or equal to $x$. By definition, we have that for any $x,y\in\RR$ and $\Delta>0$
\begin{align}
\left[[x]\bmod\Delta +y\right]\bmod\Delta=[x+y]\bmod\Delta.
\label{eq:moddist}
\end{align}
An $R$-bit modulo ADC with resolution parameter $\alpha$, or \emph{$(R,\alpha)$ mod-ADC}, is defined by
\begin{align}
[x]_{R,\alpha}\triangleq\left[\lfloor \alpha x\rfloor\right]\bmod 2^R\in\{0,1,\ldots,2^R-1\},\nonumber
\end{align}
where we have assumed that $2^R$ is an integer. In case $R$ itself is an integer, each sample of $[x]_{R,\alpha}$ can be represented by $R$ bits. Otherwise, we can buffer $n$ consecutive samples $[x_1]_{R,\alpha},\ldots,[x_n]_{R,\alpha}$ and represent them by $\lceil nR\rceil\leq nR+1$ bits, such that the average number of bits per sample is $\leq R+\tfrac{1}{n}$. The role of $\alpha$ here is to scale the input prior to quantization.
We can write $[x]_{R,\alpha}$ as
\begin{align}
[x]_{R,\alpha}=\left[\alpha x+\left(\lfloor \alpha x\rfloor-\alpha x\right)\right]\bmod 2^R=\left[\alpha x+z\right]\bmod 2^R.\label{eq:modExact}
\end{align}
The error term $z=\lfloor \alpha x\rfloor-\alpha x\in(-1,0]$ in~\eqref{eq:modExact} is clearly a deterministic function of $x$. Nevertheless, throughout this paper we will model this error term as additive uniform noise $Z\sim\Unif((-1,0])$ statistically independent of $x$, such that the $(R,\alpha)$ mod-ADC will be treated as a \emph{stochastic channel} with input $x$ and output $Y$, related as
\begin{align}
Y=[\alpha x+Z]\bmod 2^R.\label{eq:linearApprox}
\end{align}
The approximation of the $(R,\alpha)$ mod-ADC by the additive modulo channel~\eqref{eq:linearApprox} can be made exact via the use of \emph{subtractive dithers}. Specifically, we can use a random variable $U\sim\Unif([0,1))$, statistically independent of $x$, which we refer to as a \emph{dither}, and feed $\tilde{x}=x+U/\alpha$ to the $(R,\alpha)$ mod-ADC instead of feeding $x$. The output of the modulo ADC in this case will be
\begin{align}
[\tilde{x}]_{R,\alpha}&=\left[\alpha \tilde{x}+\left(\lfloor \alpha \tilde{x}\rfloor-\alpha \tilde{x}\right)\right]\bmod 2^R\nonumber\\
&=\left[\alpha x+U+\left(\lfloor \alpha x+U\rfloor-(\alpha x+U)\right)\right]\bmod 2^R\nonumber.
\end{align}
Subtracting $U$ from $[\tilde{x}]_{R,\alpha}$ and reducing the result modulo $2^R$, we obtain
\begin{align}
&\left[[\tilde{x}]_{R,\alpha}-U\right]\bmod 2^R\nonumber\\
&=\hspace{-1mm}\left[\left[\alpha x+U+\left(\lfloor \alpha x+U\rfloor-(\alpha x+U)\right)\right]\hspace{-1.25mm}\bmod 2^R-U\right]\hspace{-1.25mm}\bmod 2^R\nonumber\\
&=\hspace{-1mm}\left[\alpha x+\left(\lfloor \alpha x+U\rfloor-(\alpha x+U)\right)\right]\bmod 2^R,\nonumber
\end{align}
where the last equality follows from the distributive law of modulo~\eqref{eq:moddist}.
Note that for every $x\in \RR$, the random variable $Z=\lfloor \alpha x+U\rfloor-(\alpha x+U)$ is uniformly distributed over $(-1,0]$, and is therefore independent of $x$~\cite[Lemma 1]{ez04}. Thus, with subtractive dithers, the additive noise model~\eqref{eq:linearApprox} is exact. We note that even when dithering is not used, under suitable conditions this approximation is quite accurate~\cite{Gray90}.

Although the modulo operation entails loss of information in general, in many situations it is possible to unwrap it, i.e., reconstruct $\alpha x+Z$ from $Y=[\alpha x+Z]\bmod 2^R$ with high probability.\footnote{Here, the term ``high probability'' is used to state that this probability can be made as high as desired by increasing $R$. We explicitly quantify the relation between $R$ and the desired ``no-overload'' probability.} In particular, let
\begin{align}
\tilde{Y}=\left[Y+\frac{1}{2}2^R\right]\bmod 2^R-\frac{1}{2}2^R,
\label{eq:modshift}
\end{align}
and note that conditioned on the \emph{no-overload} event
\begin{align}
\OL\triangleq\left\{\alpha x+Z\in\left[-\frac{1}{2}2^R,\frac{1}{2}2^R\right)\right\},\nonumber
\end{align}
we have that $\tilde{Y}=\alpha x+Z$. Thus, if $\Pr(\OL)$ is close to $1$, the modulo operation has no effect with high probability. Note that $\Pr(\mathcal{E}_{\mathrm{OL}})=\Pr\left(|\alpha x+Z|>\frac{1}{2}2^R\right)$ is identical to the probability that a standard uniform quantizer with dynamic range (support) $2^R/\alpha$ is in \emph{overload}. Thus, when thinking of $x$ as a single observation, it is unclear what  the advantages of a modulo ADC are with respect to a traditional uniform ADC. However, as we illustrate below, the modulo ADC allows exploitation of the statistical structure of the acquired \emph{signal} in a much more efficient manner than the standard ADC.

The following lemma is proved using Chernoff's bound, and will be useful in the sequel for bounding $\Pr(\OL)$ in various scenarios.
\begin{lemma}[{\cite[Lemma 4]{oe15it},\cite[Theorem 7]{fsk13}}]
\label{lem:mixture}
Consider the random variable $Z_{\text{eff}}=\sum_{\ell=1}^{L} \alpha_\ell Z_\ell +\sum_{k=1}^K \beta_k U_k$ 
where $\left\{Z_\ell\right\}_{\ell=1}^L$ are iid Gaussian random variables with zero mean and some variance $\sigma^2_z$ and $\left\{U_k\right\}_{k=1}^K$ are iid random variables, statistically independent of $\left\{Z_\ell\right\}_{\ell=1}^L$, uniformly distributed over the interval $[-\rho/2,\rho/2)$ for some $\rho>0$. Let $\sigma^2_{\text{eff}}\triangleq\mathbb{E}(Z^2_{\text{eff}})$. Then for any $\tau\in\RR$
\begin{align}
\Pr(Z_{\text{eff}}>\tau)=\Pr(Z_{\text{eff}}<-\tau)\leq\exp\left\{-\frac{\tau^2}{2\sigma^2_{\text{eff}}}\right\}.\nonumber
\end{align}
\end{lemma}

\subsection{Modulo ADCs for Scalar Stationary Processes}
\label{subsec:timecor}

Let $\{X_n\}$ be a zero-mean discrete-time stationary Gaussian stochastic process, obtained by sampling a stationary Gaussian process $X(t)$ every $T_S$ seconds. Let
\begin{align}
Y_n=[\alpha X_n+Z_n]\bmod 2^R, \ n=1,2,\ldots\nonumber
\end{align}
be the process obtained by applying a $(R,\alpha)$ mod-ADC on the process $\{X_n\}$, where $\{Z_n\}$ is a $\Unif((-1,0])$ iid noise, and let
\begin{align}
V_n=\alpha X_n+Z_n, \ n=1,2,\ldots\nonumber
\end{align}
be its non-folded version. Our goal is to design a decoder that recovers $V_n$ from the outputs of the modulo ADC, $\{Y_n\}$, with high probability. To that end, we assume the decoder has access to $\{V_{n-1},\ldots,V_{n-p}\}$, an assumption that will be justified in the sequel, and that it knows the auto-covariance function $C_{X}[r]=\mathbb{E}[X_n X_{n-r}]$ of $\{X_n\}$. We apply the following algorithm (See also Figure~\ref{fig:temporal} for a schematic illustration):

\textbf{\textit{Inputs}}: $Y_n$,$\{V_{n-1},\ldots,V_{n-p}\}$, $\{C_X[r]\}$, $R$, $\alpha$.

\textbf{\textit{Output}}: Estimates $\hat{V}_n$, $\hat{X}_n$, for $V_n$ and $X_n$, respectively.

\textbf{\textit{Algorithm}}:
\begin{enumerate}
	\item Compute the optimal linear MMSE predictor for $V_n$ from its last $p$ samples 
	\begin{align}
	\hat{V}^p_n=\sum_{i=1}^p h_i \cdot \left(V_{n-i}+\frac{1}{2}\right)-\frac{1}{2},\label{eq:VnPred}
	\end{align}
	where $\{h_n\}$ is a $p$-tap prediction filter, computed based on $\{C_X[r]\}$ and $\alpha$, and the shift by $1/2$ compensates for $\mathbb{E}(Z_n)$.
	\item Compute 
	\begin{align}
	W_n&=[Y_n-\hat{V}_n^p]\bmod 2^R\nonumber\\
	\tilde{W}_n&=\left[W_n+\frac{1}{2}2^R\right]\bmod 2^R-\frac{1}{2}2^R.\nonumber
	\end{align}
	\item Output $\hat{V}_n=\hat{V}_n^p+\tilde{W}_n$, and $\hat{X}_n=\frac{\hat{V}_n+\tfrac{1}{2}}{\alpha}$.
\end{enumerate}

\begin{remark}
Note that $\{h_n\}$ is the $p$-tap prediction filter for the \emph{quantized} process $\{V_n\}$ from its past, rather than for $\{X_n\}$ from its past. While the loss for using the latter, instead of the former, becomes insignificant when high-resolution assumptions apply, it can be arbitrarily large for oversampled processes, for which high-resolution assumptions never hold~\cite{zke08,oe15c}. The filter coefficients $\{h_n\}$ need only be computed once, and can then be used for all times.
\end{remark}

The following proposition characterizes the performance of the algorithm above. All logarithms in this paper are taken to base $2$, unless stated otherwise.
\begin{proposition}
Let $\hat{V}_n^p$, $\hat{V}_n$ and $\hat{X}_n$ be as defined in the algorithm above, and let $\sigma^2_p=\mathbb{E}(V_n-\hat{V}_n^p)^2$. We have that
\begin{align}
\Pr(\mathcal{E}_{\mathrm{OL}_n})\triangleq \Pr(\hat{V}_n\neq V_n)\leq 2\exp\left\{-\frac{3}{2}2^{2\left(R-\frac{1}{2}\log(12\sigma^2_{p})\right)}\right\},
\end{align}
and
\begin{align}
D&=\mathbb{E}[(X_n-\hat{X}_n)^2|\OLn]\leq\frac{1}{12\alpha^2 (1-\Pr(\mathcal{E}_{\mathrm{OL}_n}))},
\end{align}
where the event $\OLn=\{\hat{V}_n=V_n\}$ is the complement of the event $\mathcal{E}_{\mathrm{OL}_n}=\{\hat{V}_n\neq V_n\}$.
\label{prop:time}
\end{proposition}

\begin{proof}
Let $E_n^p\triangleq V_n-\hat{V}_n^p$ be the $p$th order prediction error of the process $\{V_n\}$, and note that its variance $\sigma^2_{p}=\mathbb{E}(E_n^p)^2$ is invariant to $n$ due to stationarity. We have that
\begin{align}
W_n&=[Y_n-\hat{V}^p_n]\bmod 2^R\nonumber\\
&=\left[[V_n]\bmod 2^R-\hat{V}_n^p \right]\bmod 2^R\nonumber\\
&=\left[V_n-\hat{V}_n^p \right]\bmod 2^R\label{eq:WnMagic}\\
&=\left[E_n^p \right]\bmod 2^R,\nonumber
\end{align}
where equation~\eqref{eq:WnMagic} follows from the modulo distributive law~\eqref{eq:moddist}, and constitutes the key advantage of the modulo operation for exploiting temporal correlations. Note that $\tilde{W}_n\in[-\tfrac{1}{2}2^R,\tfrac{1}{2}2^R)$ is a cyclicly shifted version of $W_n\in[0,2^R)$, as in~\eqref{eq:modshift}. Therefore, conditioned on the event
\begin{align}
\OLn=\left\{|E^p_n|<\frac{1}{2}2^R\right\}\nonumber
\end{align}
we have that $\tilde{W}_n=E^p_n$. 

Note that $E^p_n$ is a zero-mean linear combination of statistically independent Gaussian and uniform random variables, such that Lemma~\ref{lem:mixture} applies, and we have that
\begin{align}
\Pr(\mathcal{E}_{\mathrm{OL}_n})&\triangleq\Pr(\tilde{W}_n\neq E^p_n)\nonumber\\
&=\Pr\left(|E^P_n|>\frac{1}{2}2^R\right)\nonumber\\
&\leq 2\exp\left\{-\frac{2^{2R}}{8\sigma^2_{p}}\right\}\nonumber\\
&=2\exp\left\{-\frac{3}{2}2^{2\left(R-\frac{1}{2}\log(12\sigma^2_{p})\right)}\right\},
\label{eq:Petime}
\end{align}
Whenever $\OLn$ occurs, we have that $\hat{V}_n=V_n$, and consequently
\begin{align}
\hat{X}_n=X_n+\frac{Z_n+\frac{1}{2}}{\alpha}\nonumber
\end{align}
and
\begin{align}
&\mathbb{E}[(X_n-\hat{X}_n)^2|\OLn]=\mathbb{E}\left[\left(\frac{Z_n+\frac{1}{2}}{\alpha}\right)^2\bigg|\OLn\right]\nonumber\\
&=\frac{1}{\alpha^2}\frac{\mathbb{E}(Z_n+1/2)^2-\Pr(\mathcal{E}_{\mathrm{OL}_n})\mathbb{E}[(Z_n+1/2)^2|\mathcal{E}_{\mathrm{OL}_n}]}{\Pr(\OLn)}\nonumber\\
&\leq\frac{1}{12\alpha^2 (1-\Pr(\mathcal{E}_{\mathrm{OL}_n}))}.\label{eq:Dtime}
\end{align}
\end{proof}

Proposition~\ref{prop:time} shows that we can make $\Pr(\mathcal{E}_{\mathrm{OL}_n})$ as small as $2e^{-\frac{3}{2}2^{2\delta}}$ by choosing
\begin{align}
R=\frac{1}{2}\log(12\sigma^2_{p})+\delta.
\label{eq:Rdelta}
\end{align}
For example, taking $\delta=2$ bits, results in an overload probability smaller than $10^{-10}$. In particular, unless we take a very small $\delta$, we have that $1-\Pr(\mathcal{E}_{\mathrm{OL}_n})\approx 1$, and consequently, by Proposition~\ref{prop:time}, we will have $D\approx 1/12\alpha^2$. Thus, to simplify expressions in the analysis that follows, we assume $D=1/12\alpha^2$.
We note the tradeoff in choosing $\alpha$: on the one hand, increasing $\alpha$ decreases the MSE distortion $D$, but on the other hand the prediction error variance $\sigma^2_{p}$ of the process $V_n=\alpha X_n+Z_n$ increases with $\alpha$ such that the required rate $R$ for avoiding overload errors increases. Thus, the tradeoff between $D$ and the required quantization rate is controlled through the parameter $\alpha$. We now turn to characterize the tradeoff the developed scheme achieves.

Let $h(A)$ denote the differential entropy of the random variable $A$, and $h(A|B)$ the conditional differential entropy of $A$ given the random variable $B$~\cite{coverthomas}. Recall that for a stationary Gaussian process $\{X_n\}$ with PSD $S_X(e^{j\omega})$ we have that~\cite{gray2006toeplitz}
\begin{align}
h(X_n|X_{n-1},\ldots)=\frac{1}{2\pi}\int_{\pi}^{\pi} \frac{1}{2}\log\left(2\pi e S_X(e^{j\omega})\right)d\omega,
\end{align} 
and in particular $h(X_n|X_{n-1},\ldots)=-\infty$ if and only if $S_X(e^{j\omega})=0$ over a measurable subset of $[-\pi,\pi)$.  Shannon's lower bound~\cite{berger71}, states that the number of bits per sample $R$ produced by any quantizer that attains an MSE distortion $D$ must satisfy
\begin{align}
R(D)\geq R_{\text{SLB}}(D)\triangleq h(X_n|X_{n-1},\ldots)-\frac{1}{2}\log(2\pi e D).\nonumber
\end{align}
It is well-known that for Gaussian processes with finite $h(X_n|X_{n-1},\ldots)$, Shannon's lower bound is asymptotically tight, i.e., $\lim_{D\to 0} R(D)-R_{\text{SLB}}(D)=0$,~\cite{berger71}.

\begin{proposition}
If $h(X_n|X_{n-1},\ldots)>-\infty$, then 
\begin{align}
\lim_{D\to 0}\lim_{p\to\infty} \frac{1}{2}\log(12\sigma^2_{p})=R_{\text{SLB}}(D).\nonumber
\end{align} 
\label{prob:SLB}
\end{proposition}

\begin{proof}
We can write
\begin{align}
\frac{1}{2} \log(12\sigma_{p}^2)&=\frac{1}{2}\log\left(\frac{\frac{\sigma_{p}^2}{\alpha^2}}{\frac{1}{12\alpha^2}} \right)= \frac{1}{2}\log\left(\frac{\mathbb{E}(E_n'^p)^2}{D}\right).
\end{align}
where $E_n'^p$ is the $p$th order prediction error of the process $X_n+Z_n/\alpha=X_n+\sqrt{D} \tilde{Z}_n$, where $\tilde{Z}_n\sim\Unif([-\sqrt{12},0))$ iid. 

For a Gaussian process $\{X_n\}$, the condition $h(X_n|X_{n-1},\ldots)>-\infty$ is equivalent to
\begin{align}
\frac{1}{2\pi}\int_{-\pi}^{\pi} \frac{1}{2}\log\left( S_X(e^{j\omega})\right)d\omega >-\infty.\label{eq:finiteEnt}
\end{align}
As a consequence of~\eqref{eq:finiteEnt}, we have that 
\begin{align}
\lim_{D\to 0}\frac{1}{2\pi}&\int_{-\pi}^{\pi} \frac{1}{2}\log\left(2\pi e \left( S_X(e^{j\omega})+D\right)\right)d\omega \nonumber\\
&=h(X_n|X_{n-1},\ldots).\label{eq:highSNRent}
\end{align}
By Paley-Wiener's theorem~\cite{gershogray}, we have that
\begin{align}
\lim_{p\to\infty}\mathbb{E}(E_n'^p)^2=2^{\frac{1}{2\pi}\int_{-\pi}^{\pi} \log\left( S_X(e^{j\omega})+D\right)d\omega }.\label{eq:PredErrorLim}
\end{align}
Combining~\eqref{eq:highSNRent} and~\eqref{eq:PredErrorLim}, we obtain that
\begin{align}
\lim_{D\to 0}\lim_{p\to\infty}\mathbb{E}(E_n'^p)^2=2^{2h(X_n|X_{n-1},\ldots)-2\pi e},\nonumber
\end{align}
for processes with finite entropy rate $h(X_n|X_{n-1},\ldots)$. The result now follows by rearranging terms.
\end{proof}

For the practically important case where $\{X_n\}$ is obtained by oversampling the process $\{X(t)\}$, which is studied in Section~\ref{sec:oversampled}, the assumption $h(X_n|X_{n-1},\ldots)>-\infty$ of Proposition~\ref{prob:SLB} does not hold. Nevertheless, we will show that the modulo ADC nevertheless achieves performance that is close to the information theoretic limits.

Above, we have assumed that the decoder has access to the non-folded samples $\{V_{n-1},\ldots,V_{n-p}\}$. To justify this assumption, an \emph{initialization} step is needed, where the decoder acquires the first $p$ consecutive samples $\{V_{1},\ldots,V_{p}\}$, or estimates of these samples. Once those are obtained, we can apply the algorithm described above, sample-by-sample, and assume the estimate $\hat{V}_n$ produced by the algorithm at time $n$ is correct, and can be used as an input for the algorithm in the next $p$ steps. All samples $V_{p+1},\ldots,V_T$ will be recovered correctly, as long as no overload error occurred within the $T-p$ decoding steps. Thus, by the union bound, we see that the first $T-p$ samples are recovered correctly with probability at least $1-2T e^{-\frac{3}{2}2^{2\delta}}$.\footnote{Note that conditioning on the event that no overload error occurred until time $n$, changes the statistics of $E_{n}^p$. Thus, applying the union bound correctly here requires some more care. See~\cite{oe15c} for more details.}

One conceptually simple way of performing the initialization, i.e., obtaining $\{V_{1},\ldots,V_{p}\}$ is by using a standard scalar quantizer with high-rate for the first $p$ samples. Although the high power consumption of such a quantizer will have a negligible effect on the total power consumption, due to the fact it is used only for a small fraction of the time, this approach has the disadvantage of having to include two ADCs, a high-rate standard ADC and a modulo ADC withing the system. Alternatively, one can perform the initialization using only a $R$ bit modulo ADC in one of the two following ways:
\begin{enumerate}
	\item Increase $\alpha$ gradually until it reaches its final value. For the first sample, $\alpha_1$ will be chosen such that $V_1=\alpha_1 X_1+Z_1$ is w.h.p. within the modulo interval, such that no prediction is needed. Next, we can use $V_1$ in order to predict $V_2=\alpha_2 X_2+Z_2$, which allows to use $\alpha_2>\alpha_1$ such that the prediction error is still within the modulo interval. Continuing this way, we can keep increasing $\alpha$ until convergence.
	\item We can collect a long vector of outputs from the modulo ADC, say $\{Y_1,\ldots,Y_K\}$, $K>p$, and unwrap the modulo operation via the integer-forcing source coding scheme described in the next subsection. The amount of computations per sample required in this method is greater than that of the ``steady state'', i.e., after initialization is complete, but since initialization is rarely performed, the effect on the total complexity is negligible. 
\end{enumerate}

\subsection{Modulo ADCs for Random Vectors}
\label{subsec:spacecor}

Let $\bX\sim\mathcal{N}(\mathbf{0},\bm{\Sigma})$ be a $K$-dimensional Gaussian random vector with zero mean and covariance matrix $\bm{\Sigma}$. Let
\begin{align}
Y_k=[\alpha X_k+Z_k]\bmod 2^R, \ k=1,\ldots K,\nonumber
\end{align}
be obtained by applying $K$ identical $(R,\alpha)$ mod-ADCs, each applied to a different coordinate of the vector $\bX$, where the quantization noises $Z_k\sim\Unif((-1,0])$, $k=1,\ldots,K$, are iid, and let
\begin{align}
V_k=\alpha X_k+Z_k, \ k=1,\ldots K,\nonumber
\end{align}
be its non-folded version. Our goal is to recover $\bV\triangleq[V_1,\ldots,V_K]^T$ from the outputs $\bY\triangleq[Y_1,\ldots,Y_K]^T$ of the modulo ADCs with high probability.

By definition of the modulo operation, we have that $\bV\in \bY+2^R\cdot\ZZ^K$. Consequently, the optimal decoder for $\bV$ from the measurement $\bY$, in terms of minimizing $\Pr(\hat{\bV}\neq\bV)$, is
\begin{align}
\hat{\bV}_{\text{opt}}=\bY+\argmax_{\mathbf{b}\in 2^R\cdot\ZZ^K}f_{\bV}(\bY+\mathbf{b}),\label{eq:VoptSpace}
\end{align}
where $f_{\bV}$ is the probability density function (PDF) of the random vector $\bV$. Although $f_{\bV}$ can be expressed as the convolution of the $K$-dimensional Gaussian PDF of $\alpha\bX$ and the cubic PDF of $\bZ=[Z_1,\ldots,Z_K]^T$, no simpler closed-form expression is known for it. However, as $\alpha$ increases (high-resolution quantization regime), $f_{\bV}$ approaches the pdf of a $\mathcal{N}(\mathbf{-\frac{1}{2}},\alpha^2\bm{\Sigma} +\tfrac{1}{12}\bI)$ random vector, where $\mathbf{\tfrac{1}{2}}$ is a $K$-dimensional vector with all entries equal to $\tfrac{1}{2}$ and $\bI$ is the identity matrix. Consequently, one can use the sub-optimal (in terms of minimizing $\Pr(\hat{\bV}\neq\bV)$) decoder
\begin{align}
&\hat{\bV}_{\text{Gauss}}=\bY\nonumber\\
&+\argmin_{\mathbf{b}\in 2^R\cdot\ZZ^K}(\bY+\mathbf{\frac{1}{2}}+\mathbf{b})^T\left(\alpha^2\bm{\Sigma}+\frac{1}{12}\bI\right)^{-1}(\bY+\mathbf{\frac{1}{2}}+\mathbf{b}).\nonumber
\end{align}
The matrix $(\alpha^2\bm{\Sigma}+\frac{1}{12}\bI)^{-1}$ is positive definite and therefore admits a Cholesky decomposition $\left(\alpha^2\bm{\Sigma}+\tfrac{1}{12}\bI\right)^{-1}=\bL\bL^T$
where $\bL$ is a lower triangular matrix with strictly positive diagonal entries. Setting $\bs=-\bL^T(\bf{\tfrac{1}{2}}+\bY)$, we can write
\begin{align}
\hat{\bV}_{\text{Gauss}}=\bY+2^R\cdot\argmin_{\mathbf{b}\in\ZZ^K}\left\|\bL^T \bb-2^{-R}\cdot\bs \right\|.\label{eq:clp}
\end{align}
Thus, the problem of finding $\hat{\bV}_{\text{Gauss}}$ is equivalent to that of finding the closest point to $2^{-R}\cdot\bs$ in the lattice generated by the basis $\bL^T$. Solving this problem, in general, is known to require running time exponential in $K$~\cite{mg02}, unless P=NP. Thus, for large $K$, finding $\hat{\bV}_{\text{Gauss}}$ is computationally prohibitive. One therefore needs to seek an alternative, low-complexity, decoder for $\bV$ from $\bY$. Next, we review such a decoder, proposed in~\cite{oe13b}, dubbed the \emph{integer-forcing} (IF) source decoder, see Figure~\ref{fig:modADCspatial}. The decoding algorithm works as follows.

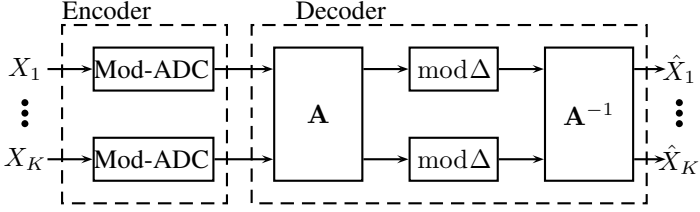
\begin{figure}[t]
\begin{center}
\psset{unit=0.6mm}
\begin{pspicture}(0,0)(250,40)

\rput(0,30){$X_{1}$}\psline{->}(5,30)(15,30)
\psframe(15,25)(42,35)\rput(28,30){Mod-ADC}
\psline{->}(42,30)(55,30)

	\rput(0,22){\Huge$\vdots$}

\rput(0,-20){
\rput(0,30){$X_{K}$}\psline{->}(5,30)(15,30)
\psframe(15,25)(42,35)\rput(28,30){Mod-ADC}
\psline{->}(42,30)(55,30)
	}

\rput(-25,0){
\psframe(80,5)(100,35)\rput(90,20){$\bA$}
\psline{->}(100,30)(110,30)\psframe(110,25)(130,35)\rput(120,30){$\bmod\Delta$}\psline{->}(130,30)(140,30)
\psline{->}(100,10)(110,10)\psframe(110,5)(130,15)\rput(120,10){$\bmod\Delta$}\psline{->}(130,10)(140,10)
\psframe(140,5)(160,35)\rput(150,20){$\bA^{-1}$}
\psline{->}(160,30)(167,30)\rput(170,30){$\hat{X}_1$}
\psline{->}(160,10)(167,10)\rput(170,10){$\hat{X}_K$}
\rput(170,22){\Huge$\vdots$}
}

\psframe[linestyle=dashed](8,0)(45,40)\rput(18,43){Encoder}

\psframe[linestyle=dashed](50,0)(138,40)\rput(70,43){Decoder}

\end{pspicture}
\end{center}
\caption{Schematic architecture for modulo ADC for random vectors.}
\label{fig:modADCspatial}
\end{figure} 

\textbf{\textit{Inputs}}: $\bY$, $\bm{\Sigma}$, $R$, $\alpha$.

\textbf{\textit{Output}}: Estimates $\hat{\bV}_{\text{IF}}$, and $\hat{\bX}_{\text{IF}}$, for $\bV$ and $\bX$, respectively.

\textbf{\textit{Algorithm}}:
\begin{enumerate}
	\item Solve
	\begin{align}
	\bA&=[\ba_1|\cdots|\ba_K]^T\nonumber\\
	&=\argmin_{\substack{{\bar{\bA}\in\ZZ^{K\times K}}\\{|\bar{\bA}|\neq 0}}}\max_{k=1,\ldots,K} \frac{1}{2}\log\left(\bar{\ba}_k^T\left(\bI+12\alpha^2\bm{\Sigma}\right)\bar{\ba}_k\right),\label{eq:Aopt}
	\end{align}
	where $|\bA|$ denotes the absolute value of $\det(\bA)$.\label{step1IFSC}
	\item For $k=1,\ldots,K$, compute
	\begin{align}
	\bar{g}_k&\triangleq \left[\ba_k^T\left(\bY+\mathbf{\frac{1}{2}}\right)\right]\bmod 2^R=[g_k]\bmod 2^R,\label{eq:idenimp}\\
	\tilde{g}_k&\triangleq\left[\bar{g}_k+\frac{1}{2}2^R\right]\bmod 2^R-\frac{1}{2}2^R,\nonumber
	\end{align}
	and set $\tilde{\bg}=[\tilde{g}_1,\ldots,\tilde{g}_K]^T$.\label{stepComp}
	\item Output $\hat{\bV}_{\text{IF}}=\bA^{-1} \tilde{\bg}$, and $\hat{\bX}_{\text{IF}}=\frac{\hat{\bV}_{\text{IF}}}{\alpha}$.
\end{enumerate}

\begin{remark}
The optimization problem~\eqref{eq:Aopt} requires a computational complexity exponential in $K$, in general (unless P=NP). However, the problem of finding the optimal integer matrix $\bA$, need only be solved once for each covariance matrix $\bm{\Sigma}$ and $\alpha$. Thus, even if the solution to this problem is computationally expensive, its cost is normalized by the number of times this solution is used. In practice, one can apply the LLL algorithm~\cite{lll82} in order to obtain a sub-optimal $\bA$ with polynomial complexity in $K$.
\end{remark}

The next proposition, adapted from~\cite[Theorem 2]{oe13b} characterizes the performance of modulo ADCs with the decoder above.
\begin{proposition}
Let $\bA=[\ba_1|\cdots|\ba_K]^T$ be the matrix found in step~\ref{step1IFSC} of the algorithm above, and define \begin{align}
R_{\text{IFSC}}(\bA)=\max_{k=1,\ldots,K}\frac{1}{2}\log\left(\ba_k^T\left(\bI+12\alpha^2\bm{\Sigma}\right)\ba_K^T\right). 
\end{align}
We have that
\begin{align}
\Pr(\mathcal{E}_{\mathrm{OL}})=\Pr(\hat{\bV}_{\text{IF}}\neq \bV)\leq 2K\exp\left\{-\frac{3}{2}\cdot 2^{2\left(R-R_{\text{IFSC}}(\bA)\right)} \right\},\nonumber
\end{align}
and
\begin{align}
D_k&=\mathbb{E}\left[\left(X_k-\hat{X}_{k,\text{IF}}\right)^2\bigg|\OL\right]\leq \frac{1}{12\alpha^2 (1-\Pr(\mathcal{E}_{\mathrm{OL}}))},\nonumber
\end{align}
for all $k=1,\ldots,K$, where the event $\OL=\{\hat{\bV}_{\text{IF}}= \bV\}$ is the complement of the event $\mathcal{E}_{\mathrm{OL}}=\{\hat{\bV}_{\text{IF}}\neq \bV\}$.
\label{prop:IFSC}
\end{proposition}

The main idea behind the decoder above is the simple observation that for any vector $\ba=[a_1,\ldots,a_k]^T\in\ZZ^K$ and any vector $\bh=[h_1,\ldots,h_K]^T\in\RR^K$ we have that
\begin{align}
\left[\sum_{k=1}^K a_k [h_k]\bmod 2^R\right]\bmod 2^R=\left[\sum_{k=1}^K a_k h_k\right]\bmod 2^R.\label{eq:iden}
\end{align}

\begin{proof}
By the identity~\eqref{eq:iden}, we have that the quantities $\bar{g}_k$, computed in step~\ref{stepComp} of the algorithm, satisfy	
\begin{align}
\bar{g}_k&= \left[\ba_k^T\left(\bY+\mathbf{\frac{1}{2}}\right)\right]\bmod 2^R=[g_k]\bmod 2^R,\nonumber
\end{align}
where
\begin{align}
g_k\triangleq \ba_k^T\left(\bV+\mathbf{\frac{1}{2}}\right)\nonumber.
\end{align}
Furthermore, $\tilde{g}_k\in[-\tfrac{1}{2}2^R,\tfrac{1}{2}2^R)$ is merely a cyclicly shifted version of $\bar{g}_k\in[0,2^R)$. Thus, $\tilde{g}_k=g_k$ if and only if $g_k\in[-\tfrac{1}{2}2^R,\tfrac{1}{2}2^R)$. Consequently, $\hat{\bV}_{\text{IF}}\neq\bV$ if and only if the event
\begin{align}
\mathcal{E}_{\mathrm{OL}}=\bigcup_{k=1}^K\left\{|g_k|\geq \frac{1}{2} 2^R \right\},\nonumber
\end{align}
occurs. Thus, by the union bound,
\begin{align}
\Pr(\mathcal{E}_{\mathrm{OL}})= \Pr(\hat{\bV}_{\text{IF}}\neq\bV)\leq \sum_{k=1}^K \Pr\left(|g_k|\geq \frac{1}{2}2^R\right).\label{eq:PeIFSCunion}
\end{align}
The random variable $g_k$ has zero mean, variance $\sigma^2_k=\ba_k^T\left(\alpha^2\bm{\Sigma}+\frac{1}{12}\bI\right)\ba_k$, and satisfies the conditions of Lemma~\ref{lem:mixture}. We therefore have that
\begin{align}
\Pr\bigg(|g_k|&\geq \frac{1}{2}2^R\bigg)\leq 2\exp\left\{-\frac{2^{2R}}{8\sigma^2_k} \right\}\nonumber\\
&=2\exp\left\{-\frac{3}{2}\cdot 2^{2\left(R-\frac{1}{2}\log(12\sigma^2_k)\right)} \right\}\nonumber\\
&=2\exp\left\{-\frac{3}{2}\cdot 2^{2\left(R-\frac{1}{2}\log\left(\ba_k^T\left(\bI+12\alpha^2\bm{\Sigma}\right)\ba_k\right)\right)} \right\}.\nonumber
\end{align}
Substituting this into~\eqref{eq:PeIFSCunion} and recalling the definition of $R_{\text{IFSC}}(\bA)$, gives
\begin{align}
P_e\leq 2K \exp\left\{-\frac{3}{2}\cdot 2^{2\left(R-R_{\text{IFSC}}(\bA)\right)}\right\}.\label{eq:PeIFSC}
\end{align}
Conditioned on the event $\OL$, i.e., the event that $\mathcal{E}_{\mathrm{OL}}$ did not occur, we have that for all $k=1,\ldots,K$
\begin{align}
D_k&=\mathbb{E}\left[\left(X_k-\hat{X}_{k,\text{IF}}\right)^2\bigg|\OL\right]\nonumber\\
&=\mathbb{E}\left[\left(\frac{Z_k+\frac{1}{2}}{\alpha}\right)^2\bigg|\OL\right]\leq \frac{1}{12\alpha^2 (1-\Pr(\mathcal{E}_{\mathrm{OL}}))},\nonumber
\end{align}
where the last inequality follows similarly to~\eqref{eq:Dtime}. 
\end{proof}

As in the previous subsection, we set 
\begin{align}
R=R_{\text{IFSC}}(A)+\delta,\label{eq:Rspace}
\end{align}
such that 
\begin{align}
\Pr(\mathcal{E}_{\mathrm{OL}})\leq 2K \exp\left\{-\frac{3}{2}\cdot 2^{2\delta}\right\},\label{eq:PeSpace}
\end{align}
and set $D=1/12\alpha^2$, which is a good approximation for the upper bound we derived on $D_k$, provided that $\delta$ is not too small. Consequently, we can write
\begin{align}
R_{\text{IFSC}}(\bA,D)&\triangleq\max_{k=1,\ldots,K} \frac{1}{2}\log\left(\ba_k^T\left(\bI+\frac{1}{D}\bm{\Sigma}\right)\ba_k\right).\label{eq:IFSC}
\end{align}

The tradeoff between rate, distortion and error probability achieved by the $(R,\alpha)$ mod-ADC with an integer-forcing decoder is therefore characterized by equations \eqref{eq:Rspace},~\eqref{eq:PeSpace},  and~\eqref{eq:IFSC}. To put this result in context, we recall the information theoretic benchmark~\cite{oe13b}
\begin{align}
R^{\text{BT}}_{\text{bench}}(D)\triangleq \frac{1}{2K}\log\left|
\bI+\frac{1}{D}\bm{\Sigma}
\right|,\nonumber
\end{align}
that approximates the minimal quantization rate, per quantizer, required by any computationally and delay unlimited system in order to achieve MSE of at most $D$ in the reconstructions of each $X_k$, $k=1,\ldots,K$.
Thus,
\begin{align}
&R_{\text{IFSC}}(\bA,D)-R^{\text{BT}}_{\text{bench}}(D)\nonumber\\
&=\frac{1}{2}\log\left(\frac{\max_{k=1,\ldots,K}\ba_k^T\left(\bI+\frac{1}{D}\bm{\Sigma}\right)\ba_k}{\left|\bI+\frac{1}{D}\bm{\Sigma} \right|^\frac{1}{K}} \right).\label{eq:spaceGap}
\end{align}
It is easy to show that the right hand side of~\eqref{eq:spaceGap} is non-negative~\cite{oe13b}. However, typically the gap is quite small, and under certain distributions of practical interest on $\bm{\Sigma}$, the cumulative distribution function (CDF) of this gap can be characterized~\cite{de17}. A comprehensive comparison between $R_{\text{IFSC}}(D)$ and $R^{\text{BT}}_{\text{bench}}(D)$ was performed in~\cite{oe13b}, and it was demonstrated that they are usually quite close. 

We further remark that the integer-forcing source decoder is merely one sub-optimal algorithm for solving~\eqref{eq:clp}. It is an interesting avenue for future research to develop alternative algorithms, with firm performance guarantees (as in~\eqref{eq:Rspace},~\eqref{eq:PeSpace} and~\eqref{eq:IFSC}), for the same problem, or more ambitiously, for solving~\eqref{eq:VoptSpace}.

\section{Oversampled Modulo-ADC}
\label{sec:oversampled}

In Section~\ref{subsec:timecor} we have demonstrated the effectiveness of the modulo ADC architecture  for acquiring stochastic processes that are correlated in time. In particular, we have shown that the performance of a modulo ADC depends on the variance of the prediction error of the process $\{V_n=\alpha X_n+Z_n\}$, rather than the variance of $V_n$ itself. However, when designing an ADC, it is desirable to impose as few constraints as possible on the signals that will be fed to the ADC.  Therefore, assuming that $\{X_n\}$ is such that $\{V_n\}$ is highly predictable may be too restrictive.

Nevertheless, recalling that the process $\{X_n\}$ is obtained by sampling a continuous-time process $X(t)$, we observe that if the sampling rate is higher than Nyquist's rate, $\{X_n\}$ will be bandlimited,\footnote{We say that a discrete-time process $\{X_n\}$ is bandlimited, if there exists some $\gamma<\pi$ such that $S_X(e^{j\omega})=0$ for all $\omega\in(-\pi,-\gamma)\cup (\gamma,\pi)$. Since our analysis takes quantization noise into account, it is quite robust to slight deviations from the assumption that $S_X(e^{j\omega})$ is strictly band limited. In particular, as long as $S_X(e^{j\omega})\ll D$, for all $\omega\in(-\pi,-\gamma)\cup (\gamma,\pi)$, where $D$ is the target MSE distortion, our analysis remains valid.} and consequently, $\{V_n\}$ will be highly predictable no matter what the precise PSD of $\{X_n\}$ happens to be. In fact, this observation can be viewed as the rationale underlying $\Sigma\Delta$-conversion. In particular, a $\Sigma\Delta$-converter is information theoretically equivalent to a differential pulse-code modulator (DPCM) whose input is a bandlimited signal with flat spectrum~\cite{oe15c}.

While having many advantages, the implementation of $\Sigma\Delta$ converters is more involved than that of traditional scalar uniform quantizers. The main challenge in the design of $\Sigma\Delta$ converters is the need to produce the quantization error, and then apply a filter to this analog signal. A major obstacle is that the generation of the quantization error requires to first quantize the current sample, then apply a digital-to-analog converter (DAC) to produce the analog representation of the quantizer's output, and finally to subtract this representation from the original sample. See Figure~\ref{fig:SigDel}. The quantizer and the DAC need to be matched as otherwise the produced quantization error is inaccurate. This, however,  turns out to be quite difficult to achieve, unless the quantizer is a simple sign detector ($1$-bit quantizer).

To circumvent the challenges listed above, we develop an oversampled modulo ADC, as an alternative to $\Sigma\Delta$-conversion. The \emph{only} assumptions made on the input process $\{X(t)\}$ is that it is bandlimited with maximal frequency at most $B$, and that its variance is at most $\sigma^2$. The developed universal architecture is as follows. See Figure~\ref{fig:temporal}.

\textbf{\textit{Analog-to-digital conversion}}: The process $X(t)$ is uniformly sampled every $T_S=1/2LB$ seconds, $L>1$, such that the sampling rate is $L$ times above Nyquist's rate. Each sample of the obtained discrete-time process $\{X_n\}$ is then discretized using an $(R,\alpha)$ mod-ADC, resulting in the quantized process $\{Y_n=[\alpha X_n+Z_n]\bmod 2^R\}$.

As above, we define the unfolded process $\{V_n=\alpha X_n+Z_n\}$. The decoding procedure assumes $\{V_{n-1},\ldots,V_{n-p}\}$ are given, and computes an estimate for $V_n$, based on $Y_n$.

\textbf{\textit{Inputs}}: $Y_n$, $\{V_{n-1},\ldots,V_{n-p}\}$, $\sigma^2$, $L$, $R$, $\alpha$.

\textbf{\textit{Outputs}}: Estimates $\hat{V}_n$ and $\hat{X}_n$ for $V_n$ and $X_n$, respectively.

\textbf{\textit{Algorithm}}: The algorithm is exactly the same as that in Section~\ref{subsec:timecor}, with only one difference. Here $\{C_X[r]\}$ is unknown. Thus, for the computation of the $p$-tap prediction filter $\{h_n\}$, we assume the PSD of $\{X_n\}$ is
\begin{align}
S_X(e^{j\omega})=\begin{cases}
L\sigma^2 & \omega\in \left[-\frac{\pi}{L},\frac{\pi}{L}\right)\\
0 & \omega\notin \left[-\frac{\pi}{L},\frac{\pi}{L}\right)
\end{cases},\label{eq:PDSflat}
\end{align}
even though this assumption may, and is most likely to, be wrong.

\textbf{\textit{Final post-processing}}: After collecting a long sequence of estimates $\{\hat{X}_1,\ldots,\hat{X}_N\}$ we apply a non-causal low pass filter 
\begin{align}
G(e^{j\omega})=\begin{cases}
\frac{12\alpha^2 L \sigma^2}{1+12\alpha^2 L \sigma^2} & \text{if }\omega\in\left[-\frac{\pi}{L},\frac{\pi}{L}\right]\\
0 & \text{if }\omega\notin\left[-\frac{\pi}{L},\frac{\pi}{L}\right]
\end{cases}\nonumber
\end{align}
on them, to obtain the sequence $\{\hat{X}^{\text{LPF}}_1,\ldots,\hat{X}^{\text{LPF}}_N\}$.

The advantages over $\Sigma\Delta$ conversion are clear: the only processing done in the analog domain is sampling and applying a modulo ADC, whereas all filtering operations are done digitally at the decoder. 

Proposition~\ref{prop:time} provides an upper bound on the error probability $\Pr(\mathcal{E}_{\mathrm{OL}_n})=\Pr(\hat{V}_n\neq V_n)$ in terms of $R-\frac{1}{2}\log(12\sigma^2_p)$. However, Proposition~\ref{prob:SLB}, which characterizes the scaling of $\frac{1}{2}\log(12\sigma^2_p)$ with $D$, does not apply here for two reasons. The first is that we use a mismatched prediction filter here, due to the unknown PSD of $\{X_n\}$, and the second is that whatever the exact PSD truns out to be, it is assumed to be supported on the frequency interval $[-\tfrac{\pi}{L},\tfrac{\pi}{L}]$, such that $h(X_n|X_{n-1},\ldots)=-\infty$, and the high-resolution assumption never holds. Instead, we prove the following.

\begin{proposition}
Let $\{X_n\}$ be a zero-mean stationary process with variance $\mathbb{E}(X^2_n)\leq\sigma^2$ and PSD supported in frequency interval $[-\tfrac{\pi}{L},\tfrac{\pi}{L}]$. Let $V_n=\alpha X_n+Z_n$ where $Z_n\sim\Unif([-1,0))$, and $\hat{V}_n^p$ be as in~\eqref{eq:VnPred}, where $\{h_n\}$ is the optimal linear MMSE $p$-tap prediction filter for $V_n$, from its past samples $\{V_{n-1},\ldots,V_{n-p}\}$, designed under the assumption that $S_X(e^{j\omega})$ is as in~\eqref{eq:PDSflat}. Then
\begin{align}
\lim_{p\to\infty}12\sigma^2_p\leq \left(1+12\alpha^2 L\sigma^2\right)^{\frac{1}{L}}.\nonumber
\end{align}
\label{prop:mismatchedpred}
\end{proposition}

\begin{proof}
Let
\begin{align}
S_{\tilde{V}}(e^{j\omega})=\begin{cases}
\alpha^2 L\sigma^2+1/12 & \omega\in[-\tfrac{\pi}{L},\tfrac{\pi}{L}]\\
1/12 & \omega\notin[-\tfrac{\pi}{L},\tfrac{\pi}{L}]
\end{cases},\label{eq:PSDdes}
\end{align}
and let $H_p(e^{j\omega})$ be the frequency response of the prediction filter $\{h_n\}$, which is designed with respect to~\eqref{eq:PSDdes}.  Further, let $H(e^{j\omega})=\lim_{p\to\infty}H_p(e^{j\omega})$. By the basic principles of optimal linear MMSE prediction, we have that
\begin{align}
S_{\tilde{V}}(e^{j\omega})|1-H(e^{j\omega})|^2=2^{\frac{1}{2\pi}\int_{-\pi}^{\pi}\log(S_{\tilde{V}}(e^{j\omega}))d\omega}.
\label{eq:whiteningCond}
\end{align} 
Therefore, combining~\eqref{eq:PSDdes} and~\eqref{eq:whiteningCond}, we see that
\begin{align}
|1-H(e^{j\omega})|^2=\begin{cases}
\left(1+12\alpha^2 L\sigma^2\right)^{\frac{1}{L}-1} & \omega\in[-\tfrac{\pi}{L},\tfrac{\pi}{L}]\\
\left(1+12\alpha^2 L\sigma^2\right)^{\frac{1}{L}} & \omega\notin[-\tfrac{\pi}{L},\tfrac{\pi}{L}]
\end{cases}.
\end{align}
Applying this filter on the ``actual'' process $V_n={\alpha X_n+Z_n}$, whose PSD is
\begin{align}
S_V(e^{j\omega})=\begin{cases}
\alpha^2 S_X(e^{j\omega})+1/12 & \omega\in[-\tfrac{\pi}{L},\tfrac{\pi}{L}]\\
1/12 & \omega\notin[-\tfrac{\pi}{L},\tfrac{\pi}{L}]
\end{cases},\nonumber
\end{align}
we get
\begin{align}
\lim_{p\to\infty}12\sigma^2_p&=\lim_{p\to\infty}12\mathbb{E}(V_n-V^p_n)^2\nonumber\\
&=\frac{1}{2\pi}\int_{-\pi}^{\pi}S_V(e^{j\omega})|1-H(e^{j\omega})|^2d\omega\nonumber\\
&=\frac{\left(1+12\alpha^2 L\sigma^2\right)^{\frac{1}{L}}}{2\pi}\bigg[\int_{\omega\notin[-\pi/L,\pi/L]} 1 d\omega\nonumber\\
&+\int_{-\pi/L}^{\pi/L}\left(1+12\alpha^2 L\sigma^2\right)^{-1}\left(1+12\alpha^2 S_X(e^{j\omega})\right)d\omega\bigg]\nonumber\\
&\leq \left(1+12\alpha^2 L\sigma^2\right)^{\frac{1}{L}},
\end{align}
where the last inequality follows from our assumption that $\frac{1}{2\pi}\int_{-\pi/L}^{\pi/L}S_X(e^{j\omega})d\omega=\mathbb{E}(X_n^2)\leq \sigma^2$.
\end{proof}

It follows from Proposition~\ref{prop:time} combined with Proposition~\ref{prop:mismatchedpred}, that for a quantization rate of
\begin{align}
R=\delta+\frac{1}{L}\frac{1}{2}\log\left(1+12\alpha^2 L\sigma^2 \right).\label{eq:Roversampled}
\end{align}
the proposed system achieves $\Pr(\mathcal{E}_{\mathrm{OL}_n})\leq 2\exp\{-\frac{3}{2}2^{2\delta}\}$, for all input processes with bandwidth $\leq B$ and variance $\leq \sigma^2$. 

After low-pass filtering with $G(e^{j\omega})$, we get by a similar analysis to that done in Section~\ref{subsec:timecor} and in~\cite{oe15c}, that for long enough $N$ such that the LPF can be treated as ideal, we have that
\begin{align}
D&=\mathbb{E}\left[(X_n-\hat{X}^{\text{LPF}}_n)^2 \bigg| \bigcap_{n=1}^N\{ \hat{V}_n=V_n\}\right]\nonumber\\
&\leq \frac{\sigma^2}{1+12\alpha^2 L\sigma^2}\frac{1}{1-\Pr\left(\bigcap_{n=1}^N\{ \hat{V}_n=V_n\}\right)}\nonumber\\
&\leq \frac{\sigma^2}{1+12\alpha^2 L\sigma^2}\frac{1}{1-N\Pr(\mathcal{E}_{\mathrm{OL}_n})}\nonumber\\
&\leq \frac{\sigma^2}{1+12\alpha^2 L\sigma^2}\frac{1}{1-2N\exp\{-\frac{3}{2}2^{2\delta}\}}.
\end{align}
Thus, for large enough $\delta$ such that the total overload probability is small, i.e., 
\begin{align}
2N\exp\left\{-\frac{3}{2}2^{2\delta}\right\}\ll 1,\label{eq:deltaOversampled}
\end{align}
we have that our system achieves distortion $\approx D$ with
\begin{align}
 R=\frac{1}{L}\frac{1}{2}\log\left(\frac{\sigma^2}{D}\right)+\delta.\label{eq:RDoversampled}
\end{align}
The term $\frac{1}{L}\frac{1}{2}\log(\frac{\sigma^2}{D})$ is the rate-distortion function of a source with PSD as in~\eqref{eq:PDSflat}. Thus, up to the loss of $\delta$ bits per sample, due to the one dimensional quantizer we are using, whose size is dictated by~\eqref{eq:deltaOversampled}, our system is optimal in the following minimax sense: no system can attain a better tradeoff between $R$ and $D$ simultaneously for all processes with bandwidth at most $B$ and variance at most $\sigma^2$.

The multiplicative increase in quantization rate of the developed system, with respect to the fundamental rate-distortion limit, is $(\frac{1}{2}\log\left(\frac{\sigma^2}{D}\right)+L\delta)/(\frac{1}{2}\log\left(\frac{\sigma^2}{D}\right))$. If $X(t)$ were sampled at its Nyquist rate, rather than $L$ times above it, standard uniform scalar quantization would have achieved similar overload probability and distortion with only a $(\frac{1}{2}\log\left(\frac{\sigma^2}{D}\right)+\delta)/(\frac{1}{2}\log\left(\frac{\sigma^2}{D}\right))$ multiplicative increase in rate with respect to the fundamental limit. The disadvantage of the latter approach is that it requires to use a high-resolution quantizer for each sample, whereas the scheme developed here, allows to reduce the number of quantization bits per sample, at the expanse of an increased sampling rate. Thus, just like $\Sigma\Delta$ conversion, the scheme developed here allows to replace slow but high-resolution ADCs, with fast low-resolution ones.

\section{Implementation via Ring Oscillators}
\label{sec:ringosc}

In this Section we develop an architecture for a circuit implementing a modulo ADC, and provide a mathematical model for its input-output characteristic. Our implementation is essentially based on converting the input voltage into phase, which can naturally only be observed modulo $2\pi$, and then quantizing the phase. To that end, we use \emph{ring oscillator ADCs}, as described next.

Consider a closed-loop cascade of $N$ inverters, where $N$ is an odd number, all controlled with the same voltage $V_{dd}$, see Figure~\ref{fig:ringosc}. This circuit, which is referred to as a ring oscillator can act as an ADC with sampling period $T_s$, when $V_{dd}$ is set to $V_{in}(t)=g(X(t))$, where $X(t)$ is the analog signal to be converted to a digital one and $g(\cdot)$ is a function to be specified, and the state (`$0$' or `$1$', corresponding to `low' or `high') of each inverter is measured every $T_s$ seconds.

It is well known that the time it takes for a non-ideal inverter's output to respond to a change in its input is a function of $V_{dd}$~\cite{rcn03}, which we denote by $\Delta(V_{dd})>0$. Taking this delay into account, a moment of reflection reveals that at each time instance, exactly one pair of adjacent inverters are at the same state whereas all other pairs of adjacent inverters are at distinct states. Denote by $I\in\{1,\ldots,N\}$ the index of the first inverter within the pair that shares the same state, and denote its state by $B\in\{0,1\}$, i.e., the adjacent pair of inverters with the same state are inverter $I$ and inverter $[I+1]\bmod N$, and their state is $B$. With this notation, we can uniquely identify the states of all $N$ inverters at time $t$ with the number $Q_t=(I_t-1)+N\cdot [I_t+B_t]\bmod 2\in\{0,\ldots,2N-1\}$. See Figure~\ref{fig:inverterstates}.

A crucial observation is that the process $Q_t$ cyclically oscillates in increments of $+1$ modulo $2N$. More formally stated, if $t'>t$ is the earliest time where $Q_{t'}\neq Q_t$, then $Q_{t'}=[Q_t+1]\bmod 2N$. We designate by $V_n$ the number of increments that occurred in the process $\{Q_t\}$ within the time interval $[nT_S,(n+1) T_s)$, and define the output of the induced modulo ADC as
\begin{align}
Y_{n}\triangleq [V_{n}]\bmod 2N=[Q_{(n+1)T_s}-Q_{nT_s}]\bmod 2N.\nonumber
\end{align}

Next, we relate $V_{n}$ to the process $V_{in}(t)$. To this end, we make the simplifying assumption that $X(t)$ is constant within each time interval $[nT_s,(n+1)T_s)$, and consequently, so is $V_{in}(t)$. This assumption can be made exact by adding a sample-and-hold circuit to the system. Assuming the function $\Delta(V_{dd})$ is identical for all $N$ inverters, we have that
\begin{align}
Q_{nT_s}&=\left[ \left\lfloor\sum_{k=-\infty}^{n-1}\frac{T_s}{\Delta(V_{in}(kT_s))}\right\rfloor\right]\bmod 2N,\nonumber
\end{align}
and consequently,
\begin{align}
&Y_{n}=\Bigg[\left[\left\lfloor\sum_{k=-\infty}^n \frac{T_s}{\Delta(V_{in}(kT_s))}\right\rfloor\right]\bmod 2N\nonumber\\
&~~~-\left[ \left\lfloor\sum_{k=-\infty}^{n-1}\frac{T_s}{\Delta(V_{in}(kT_s))}\right\rfloor\right]\bmod 2N\Bigg]\bmod 2N\nonumber\\
&\hspace{-0.5mm}=\left[\left\lfloor\sum_{k=-\infty}^n \frac{T_s}{\Delta(V_{in}(kT_s))}\right\rfloor \hspace{-1.25mm} - \hspace{-1.25mm} \left\lfloor\sum_{k=-\infty}^{n-1}\frac{T_s}{\Delta(V_{in}(kT_s))}\right\rfloor\right]\hspace{-1.25mm}\bmod 2N,\nonumber
\end{align}
where the last equality follows from the modulo distributive law~\eqref{eq:moddist}. Defining the ``quantization error''
\begin{align}
Z_n=\left\lfloor\sum_{k=-\infty}^n \frac{T_s}{\Delta(V_{in}(kT_s))}\right\rfloor \hspace{-1.25mm} - \hspace{-1.25mm} \sum_{k=-\infty}^n \frac{T_s}{\Delta(V_{in}(kT_s))}\in(-1,0],\nonumber
\end{align}
we can write
\begin{align}
Y_n&=\Bigg[\sum_{k=-\infty}^n \frac{T_s}{\Delta(V_{in}(kT_s))}+Z_{n}\nonumber\\
&-\sum_{k=-\infty}^{n-1} \frac{T_s}{\Delta(V_{in}(kT_s))}-Z_{n-1} \Bigg]\bmod 2N\nonumber\\
&=\left[\frac{T_s}{\Delta(V_{in}(nT_s))}+Z_{n}-Z_{n-1} \right]\bmod 2N.\nonumber
\end{align}
Let us now define the function
\begin{align}
f(x)=\frac{1}{\Delta(x)},\nonumber
\end{align}
which corresponds to the oscillation frequency of our circuit, and is dictated by the characteristics of the inverters at hand, and let us also take the function $g(\cdot)$ to be affine, such that $V_{in}(t)=a+bX(t)$. We further define the discrete time process $X_n=X(nT_s)$, for all $n\in\mathbb{N}$. We have therefore obtained the model
\begin{align}
Y_n=\left[T_s \cdot f(a+bX_n)+Z_n-Z_{n-1} \right]\bmod 2N.\label{eq:modADCringModel}
\end{align}
In general, the quantization noise process $\{Z_n\}$ is a deterministic function of the process $\{X_n\}$. Nevertheless, as in the analysis of the ideal modulo ADC, in the sequel we make the simplifying assumption that it is an iid process with $Z_n\sim\Unif((-1,0])$.

If $f(\cdot)$ were an affine function itself, with an appropriate choice of the parameters $a,b$ we could have induced the model
\begin{align}
Y_n=[\alpha X_n+Z_n-Z_{n-1}]\bmod 2^R,\nonumber
\end{align}
where $R=\log(2N)$, which is identical to the ideal $(R,\alpha)$ mod-ADC, up to the fact that the quantization noise $Z_n-Z_{n-1}$ is now a first order moving-average (MA) process rather than a white process. In practice, however, it is difficult to construct inverters for which $f(\cdot)$ is approximately affine within a large range. The effect of nonlinearities of $f(\cdot)$ on the performance of the modulo ADC is numerically studied in the next section.

\section{Numerical Experiments}
\label{sec:numerical}

We have conducted numerical simulations for the performance of a ring oscillator based modulo ADC, where the input is an oversampled process, as in Section~\ref{sec:oversampled}. In our simulations, we have assumed that the inverters were produced using a CMOS technology. The corresponding function $f(V_{\text{in}})$ relating the input voltage to the output frequency of the oscillator, which was introduced in Section~\ref{sec:ringosc}, is shown in Figure~\ref{fig:freqLUT}, as obtained using a PSpice simulation.

\begin{figure}
\includegraphics[width=0.5\textwidth]{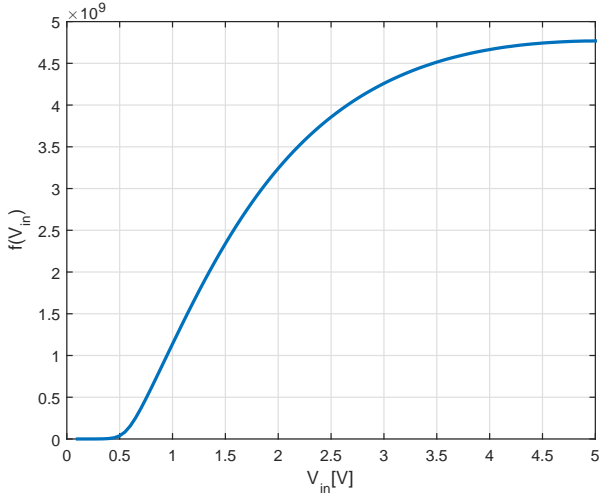}
\caption{The voltage to output frequency function $f(V_{\text{in}})$.}
\label{fig:freqLUT}
\end{figure}

\subsection{Design of System Parameters}

In all our simulations, we have designed the modulo ADC and the corresponding decoder as described in Section~\ref{sec:oversampled}, i.e., under the assumption that the input signal $X(t)$ is a Gaussian stationary process with zero mean and variance $\sigma^2$, whose PSD is flat within the frequency interval $[-B,B]$ and zero outside this interval. The sampling rate is a factor of $L>1$ above the Nyquist rate, such that the sampling period is $T_s=\tfrac{1}{2LB}$ seconds.

Given the oversampling ratio $L$, the number of inverters $N$, and the above assumptions on the statistics of $X(t)$, the design of the modulo ADC and its corresponding decoder consists of:
\begin{enumerate}
\item Choosing the shift and scaling parameters $a$ and $b$ for the modulo ADC such that $V_{\text{in}}(t)=a+bX(t)$;
\item Designing the $p$-tap prediction filter $\{h_n\}$ for $V_n=T_s f(a+bX_n)+Z_n-Z_{n-1}$ given the past samples $\{V_{n-1},\ldots,V_{n-p}\}$;
\item Designing a $2k+1$-tap noncausal smoothing filter $\{g_n\}$ for estimating $X_n$ from $\{V_{n-k},\ldots,V_{n+k}\}$.
\end{enumerate}

The decoding procedure consists of recovering an estimate $\{\hat{V}_n\}$ for $\{V_n\}$ from the modulo ADC's outputs $\{Y_n=[T_s f(a+bX_n)+Z_n-Z_{n-1}]\bmod 2N\}$, by applying the decoding procedure described in Section~\ref{sec:oversampled} with the prediction filter $\{h_n\}$. Then, the estimate $\{\hat{X}_n\}$ is produced by applying the smoothing filter $\{g_n\}$ to the process $\{\hat{V}_n\}$, which is referred to as final post-processing in Section~\ref{sec:oversampled} . The filters $\{h_n\}$ and $\{g_n\}$ are chosen as the MMSE-optimal linear prediction and smoothing filters, respectively. Calculating the coefficients of $\{h_n\}$ requires knowledge of  the second-order statistics of the process $\{V_n\}$. This in turn, can be (numerically) calculated from the pairwise distribution of $\{X_n,X_{n-m}\}$, $m=0,\ldots,p$, which is fully characterized by our assumption that $\{X_n\}$ is a Gaussian process with PSD $S_X(e^{j\omega})$ as in~\eqref{eq:PDSflat}. Calculating the coefficients of $\{g_n\}$ requires, in addition, the
joint second-order statistics of the processes $\{X_n,V_n\}$, which can either be calculated numerically, or via Bussgang's Theorem~\cite{bussgang52}.

We apply the developed modulo ADC architecture to processes of length $T$ discrete samples. The parameters $a$ and $b$ are chosen as follows: Let $P_e=\Pr(\cup_{t=1}^T \hat{V}_t\neq V_t)$ be the block error probability of our decoder, and let $\epsilon$ be our target block error probability. For every $a$ and $b$, we find the filters $\{h_n\}$ and $\{g_n\}$ as described above, and compute the corresponding $P_e=P_e(a,b)$ via Monte Carlo simulation for a Gaussian input process with PSD as in~\eqref{eq:PDSflat}. Among all $(a,b)$ for which $P_e(a,b)<\epsilon$, we choose the pair that results in the smallest MSE distortion $\frac{1}{T}\sum_{t=1}^T\mathbb{E}(X_t-\hat{X}_t)^2$. The target block error probability for all of the setups we consider is $\epsilon=10^{-3}$, and the block length we consider is $T=2^{11}$. Roughly, these parameters correspond to allowing a per-sample overload error probability of $~10^{-3}\cdot 2^{-11}\approx 4.89\cdot 10^{-7}$.

\subsection{Evaluation Method}

The system was designed for a bandlimited Gaussian process with a flat PSD. Nevertheless, we would like it to achieve approximately the same MSE distortion and error probability for all bandlimited processes with the same variance, regardless of the PSD within that band. For an ideal modulo ADC and large $p$, this is indeed the case, as shown in Section~\ref{sec:oversampled}. To test to what extent this remains the case also for the ring oscillator based modulo ADC, we apply our system on two types of processes: 1) A Gaussian process with variance $\sigma^2$ and bandwidth $B$, whose PSD is flat within this band, for which the system was designed; 2) A sinusoidal waveform, whose frequency is chosen at random, uniformly on $[0,B)$, and whose amplitude is $\sqrt{2\sigma^2}$, such that its power is $\sigma^2$. 

For each experiment, we also plot the theoretical performance of an ideal $(R,\alpha)$ mod-ADC, as well as those of a first-order $\Sigma\Delta$ (with the optimal $1$-tap noise shaping filter) converter, both designed to achieve the same target block error probability for the bandlimited Gaussian stochastic process $X(t)$. Although overload errors have a different effect on $\Sigma\Delta$ converters and modulo ADCs, both systems fail to achieve their target distortions unless those are avoided.

In the ADC literature, it is quite common to measure the performance of a particular ADC for a sinusoidal input. One drawback of this approach is that the deterministic nature of the input signal allows to design the ADC such that overload errors \emph{never} occur, without significantly increasing its dynamic range above the standard deviation of its input. For stochastic processes, even if Gaussianity is assumed, the dynamic range must be as large as multiple standard deviations of its input, in order to ensure a small overload probability. In our derivations, this is manifested through the rate backoff parameter $\delta$, which dictates the ratio between the quantizer's dynamic range $2^R$ and the standard deviation of its input (which in our case is the prediction error processes).

\begin{figure*}[!htb]
\begin{center}
\subfloat[$B=100$Hz, $L=3$]{
\includegraphics[width=0.45\textwidth]{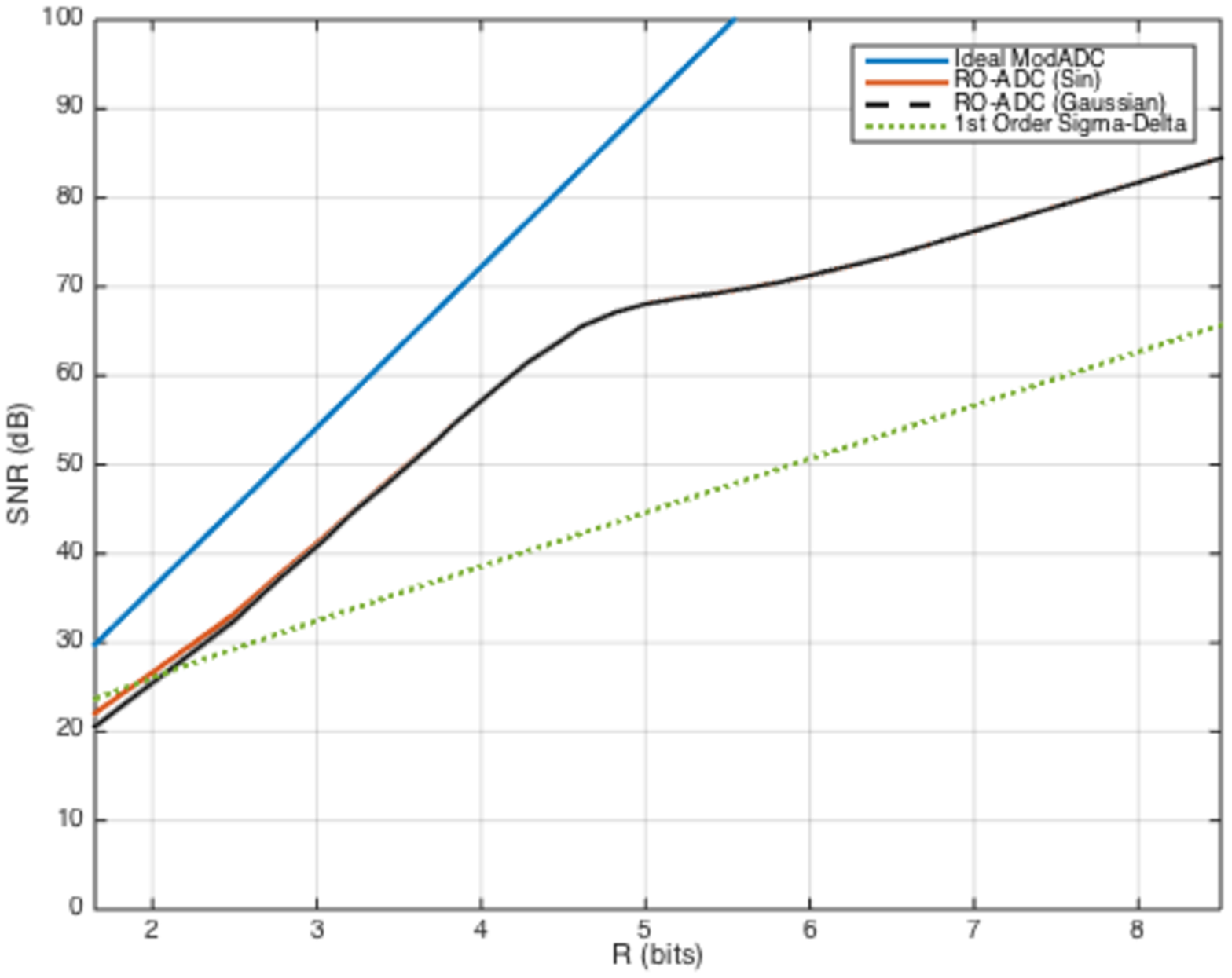}
\label{fig:100Hz}}
\qquad
\subfloat[$B=44.1$KHz, $L=3$]{
\includegraphics[width=0.45\textwidth]{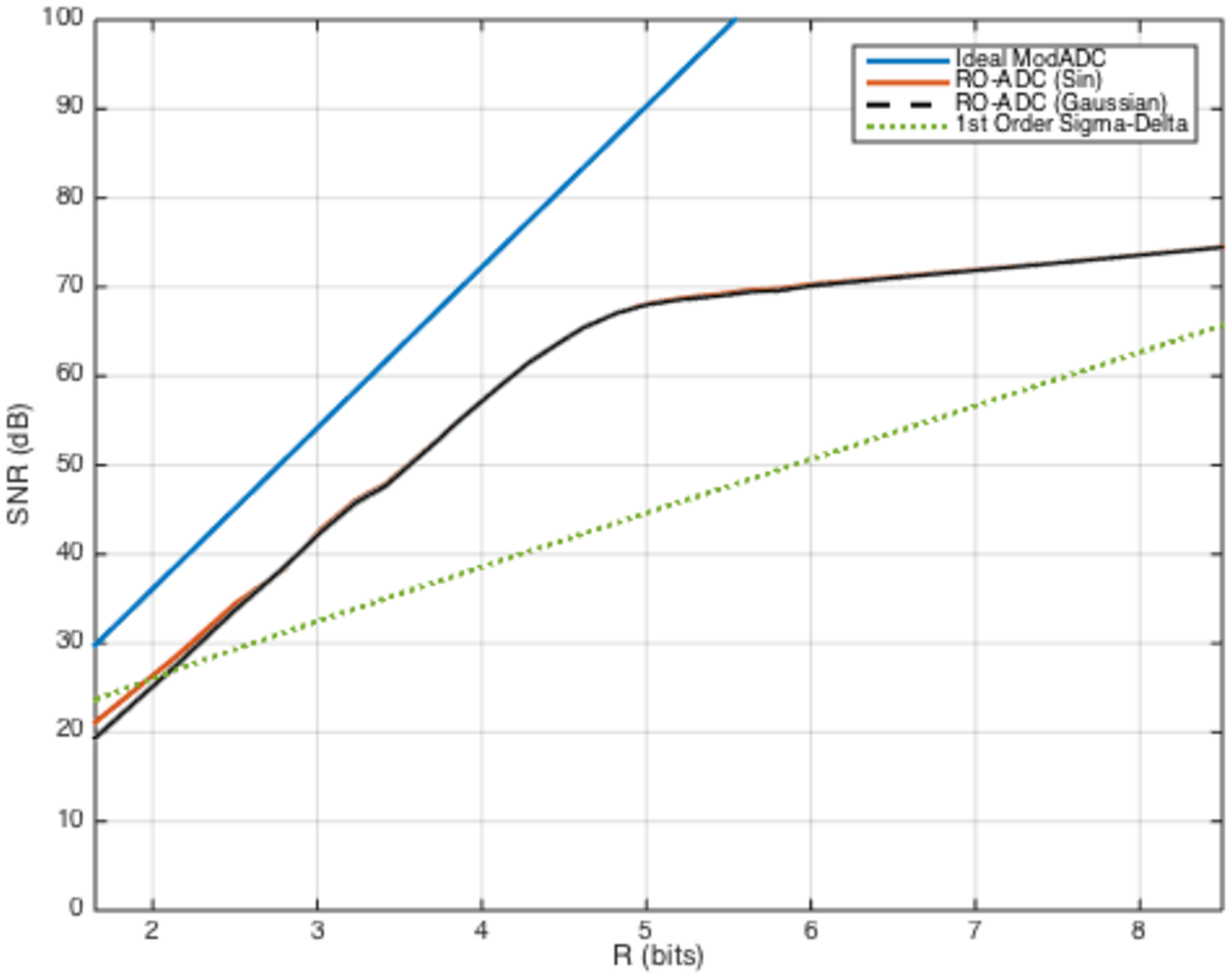}
\label{fig:44KHz}}

\subfloat[$B=100$KHz, $L=3$]{
\includegraphics[width=0.45\textwidth]{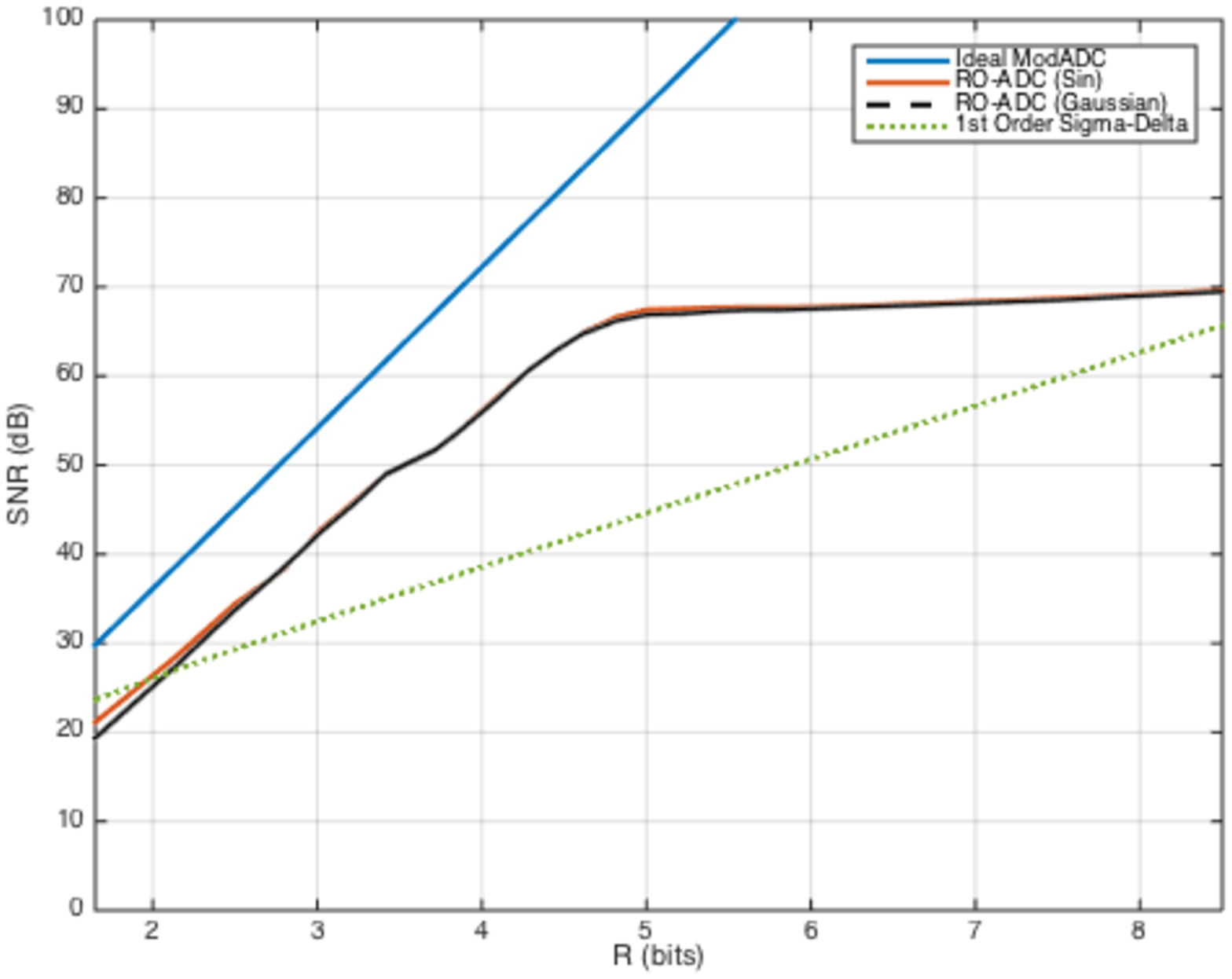}
\label{fig:100KHz}}
\qquad
\subfloat[$B=1$MHz, $L=3$]{
\includegraphics[width=0.45\textwidth]{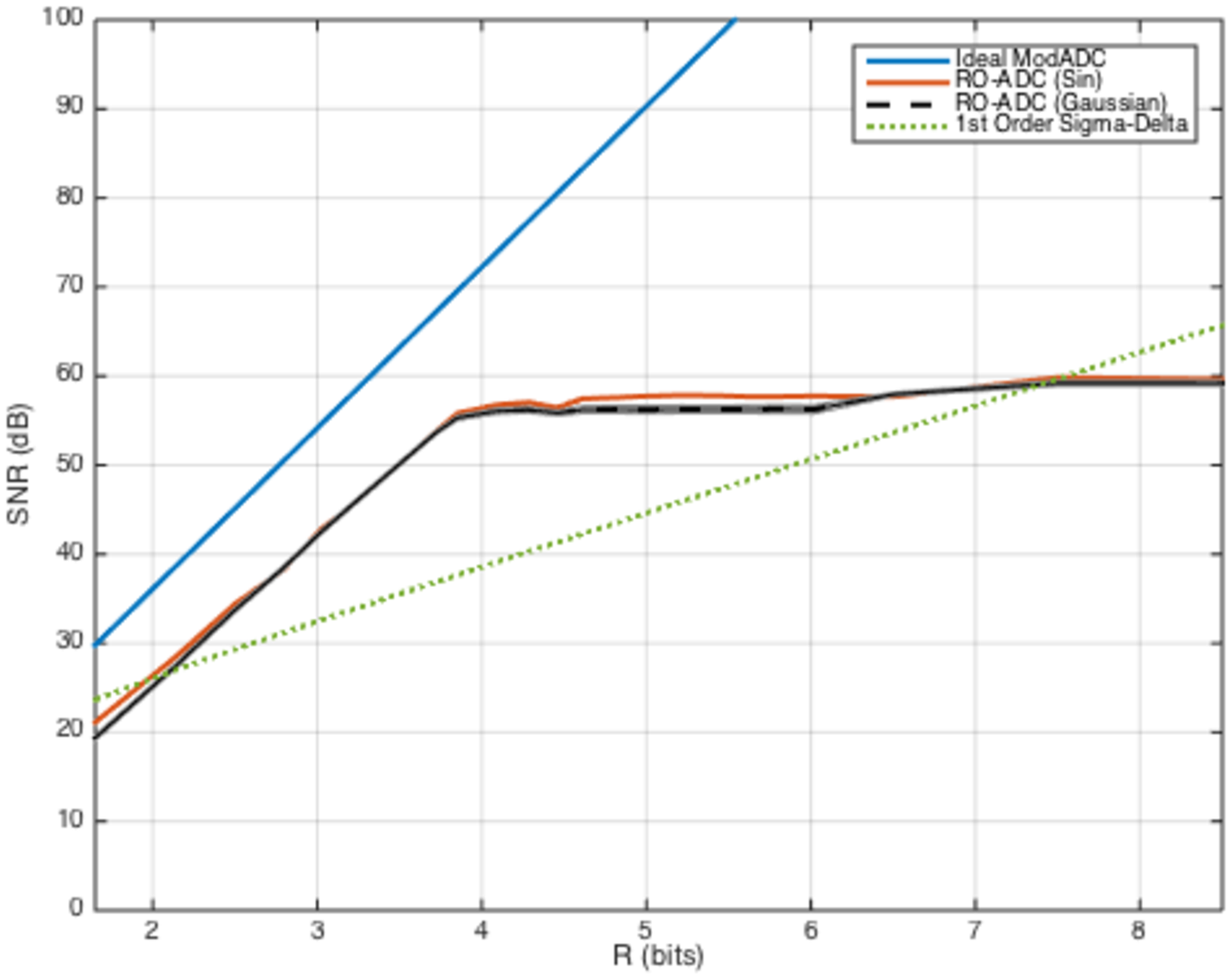}
\label{fig:1MHz}}
\end{center}
\caption{Performance of ring oscillators based modulo ADC (RO-ADC). We plot SNR vs. quantization rate for a Gaussian process and for a sinusoidal waveform processes with a random frequency, uniformly distributed over $[0,B)$. For comparison we also plot the performance of an ideal $(R,\alpha)$ mod-ADC, as well as those of an ideal first-order $\Sigma\Delta$ converter. For all curves, SNR is defined as $\sigma^2/D$. The prediction filter has $p=25$ taps, whereas the smoothing filter has $2k+1$ taps for $k=22$.}
\label{f:symICrates}
\end{figure*}

In order to allow a unified presentation of the results for both Gaussian and sinusoidal processes, rather than plotting the rate $R_{\text{mod-ADC}}(D)$ required by the modulo ADC in order to achieve an MSE distortion $D$ with target block error probability $\epsilon$, we plot $R_{\text{mod-ADC}}(D)-\delta$, where
\begin{align}
\delta=\frac{1}{2}\log\left(-\frac{2}{3}\ln\left(\frac{\epsilon}{2T}\right)\right).\label{eq:deltaNum}
\end{align}
This is consistent with traditional converter analyses that separate saturation effects from granularity ones~\cite{jayantnoll,gershogray}.
For our parameters $T=2^{11}$, $\epsilon=10^{-3}$,~\eqref{eq:deltaNum} evaluates to $\delta\approx 1.6717$ bits. Note that by~\eqref{eq:Rdelta}, $\delta$ is the rate backoff required in order to attain block error probability below $\epsilon$ by an ideal modulo ADC, when the input process is Gaussian. A similar analysis reveals that the same rate backoff is also required for a $\Sigma\Delta$ converter to attain the same block error probability, under the same assumptions on the input process~\cite{oe15c}. Thus, in all figures we also plot $R_{\Sigma\Delta}(D)-\delta$ rather than $R_{\Sigma\Delta}(D)$, where $R_{\Sigma\Delta}(D)$ is the rate needed by the $\Sigma\Delta$ converter to attain distortion $D$ with block error probability below $\epsilon$.

\subsection{Results and Discussion}

We have performed experiments for the parameters $L=3$ and four different values of $B$: $100$Hz, $44.1$KHz, $100$KHz and $1$MHz. The value of $\sigma^2$ is immaterial, as it can be absorbed in the parameter $b$. The results are depicted in Figures~\ref{fig:100Hz},~\ref{fig:44KHz},~\ref{fig:100KHz} and~\ref{fig:1MHz}, respectively. The results are based on Monte Carlo simulation, with $10^3$ independent trials for each point in each figure. No overload errors were observed for the choices of $a$, $b$, $\{h_n\}$ and $\{g_n\}$ that correspond to each point in the figures, neither for the Gaussian processes and neither for the sinusoidal processes.

In general, the results indicate that the ring oscillator implementation of a modulo ADC is closer to the ideal modulo ADC for small bandwidths $B$ and quantization rates $R$. In all figures we observe the same trend: for small enough $R$ the curve of the SNR as a function of $R$ for the ring oscillator modulo ADC is parallel to that of the ideal modulo ADC, and has a slope of $\approx 6L=18$dB/bit, in agreement with~\eqref{eq:RDoversampled}. Then, for large enough $R$ the system's non-linearities ``kick-in'' and the slope significantly decreases. Eventually, for large enough $R$, the first-order $\Sigma\Delta$ converter outperforms the ring oscillator modulo ADC, as can be observed in Figure~\ref{fig:1MHz}. Nevertheless, for moderate values of $R$, even for $B=1$MHz, the improvement over the $\Sigma\Delta$ converter can be as large as $17$dB.

The trends above are to be expected. Recall that the output of the corresponding modulo ADC is given by~\eqref{eq:modADCringModel}. If $b\cdot\sigma$ is small enough, the function $f(a+bX_n)$ resides in a small interval around $f(a)$ with high probability, and is well approximated by the linear function $f(a)+bf'(a)X_n$. Consequently, the output of the modulo ADC can be well approximated as
\begin{align}
Y_n\approx[T_s b f'(a) X_n+Z_n-Z_{n-1}+T_s f(a)]\bmod 2N.\nonumber
\end{align}
Since $T_s f(a)$ is known and can be removed, this is equivalent to a $(T_s b f'(a),\log(2N))$ mod-ADC, albeit with quantization noise $ Z_n-Z_{n-1}$ rather than $Z_n$.

Typically, however, in order to get a large gain from using a modulo ADC rather than a standard uniform quantizer, we would like to use an $(R,\alpha)$ mod-ADC with $\alpha\cdot\sigma\gg \tfrac{1}{2}2^R$. Thus, in order to get a ``useful'' modulo ADC that is close to ideal, the two conditions (i) $b\cdot\sigma\ll 1$; (ii) $T_s f'(a) \cdot b\cdot\sigma\gg N$; should hold. These two conditions can only be satisfied simultaneously if $T_s f'(a)\gg 1$, i.e., when the sampling rate is low, relative to $f'(a)$.

For an ideal $(R,\alpha)$ mod-ADC with a given target overload error probability, as $R$ increases $\alpha$ can also increase, resulting in a smaller distortion. Similarly, for the ring oscillator modulo ADC, the optimal choice of $b$ should, in general, increase with $R$. For small rates, the optimal value of $b$ is also small, such that the linear approximation for the function $f(\cdot)$ is not too bad. However, as $R$, and consequently $b$, increases, the nonlinearities start becoming significant and the slope of the SNR as a function of $R$ becomes smaller.

\section{Modulo ADCs for Jointly Stationary Processes}
\label{subsec:spacetimecor}

In this section we develop a scheme that uses $K$ parallel modulo ADCs for digitizing $K$ jointly stationary processes, provide a corresponding low-complexity decoding algorithm, and characterize its performance.

Let $\{X^1_n\},\ldots,\{X^K_n\}$ be $K$ discrete-time jointly Gaussian stationary random processes, obtained by sampling the jointly Gaussian stationary processes $X_1(t),\ldots,X_K(t)$ every $T_s$ seconds. Let
\begin{align}
Y^k_n=[\alpha X^k_n+Z^k_n]\bmod 2^R, \ k=1,\ldots,K, \ n=1,2,\ldots\nonumber
\end{align}
be the processes obtained by applying $K$ parallel $(R,\alpha)$ mod-ADCs, on $\{X^1_n\},\ldots,\{X^K_n\}$, where the input to the $k$th modulo ADC is the process $\{X^k_n\}$, and $\{Z^k_n\}$ is a $\Unif((-1,0])$ noise, iid in space and in time. Let
\begin{align}
V^k_n=\alpha X^k_n+Z^k_n, \ k=1,\ldots,K, \ n=1,2,\ldots\nonumber
\end{align}
be the non-folded version of $Y_k^n$. Let $\bX_n=[X^1_n,\ldots,X^K_n]^T$, and define $\bY_n$, $\bZ_n$ and $\bV_n$ similarly. Our goal is to recover the process $\{\bV_n\}$ from the outputs of the modulo ADCs with high probability.

To achieve this goal, we employ a two-step procedure, combining the schemes from Section~\ref{subsec:timecor} and Section~\ref{subsec:spacecor}: first we compute a predictor $\hat{\bV}^p_n$ based on previous $p$ samples $\{\bV_{n-1},\ldots,\bV_{n-p}\}$ whose error is the vector $\bE^p_n=\bV_n-\hat{\bV}^p_n$. By the same derivation as in Section~\ref{subsec:timecor}, we can produce $\left[\bE^p_n\right]\bmod 2^R$ from $\bY_n$ and $\{\bV_{n-1},\ldots,\bV_{n-p}\}$, where the modulo operation applied to a vector is to be understood as reducing each coordinate modulo $2^R$. Now, our task is to decode a modulo-folded correlated random vector, which can be done via the integer-forcing decoder described in Section~\ref{subsec:spacecor}. This relatively simple decoding procedure allows to efficiently exploit both temporal and spatial correlations. Below we describe it in more detail. See Figure~\ref{fig:spatiotemporal}. For all $\ell,m\in \{1,\ldots,K\}$, let $C_{\ell m}[r]=\mathbb{E}(X^{\ell}_n X^m_{n-r})$.

\textbf{\textit{Inputs}}: $\bY_n$, $\{\bV_{n-1},\ldots,\bV_{n-p}\}$, $\{C_{\ell m}[r]\}$ for all $\ell,m\in \{1,\ldots, K\}$, $R$, $\alpha$.

\textbf{\textit{Outputs}}: Estimates $\hat{\bV}_n$ and $\hat{\bX}_n$ for $\bV_n$ and $\bX_n$, respectively.

\textbf{\textit{Algorithm}}:
\begin{enumerate}
	\item Compute the optimal linear MMSE predictor for $\bV_n$ from its last $p$ samples 
	\begin{align}
	\hat{\bV}^p_n=\sum_{i=1}^p \bH_i \cdot \left(\bV_{n-i}+\bf{\frac{1}{2}}\right)-\bf{\frac{1}{2}},\nonumber
	\end{align}
	where $\{\bH_n\}$ is a $p$-tap matrix prediction filter, $\bH_i\in\RR^{K\times K}$, for $i=1,\ldots,p$, computed based on $\{C_{\ell m}[r]\}$ for all $\ell,m\in \{1,\ldots, K\}$ and $\alpha$, and the shift by $\bf{\tfrac{1}{2}}$ compensates for $\mathbb{E}(\bZ_n)$.
	\item Compute 
	\begin{align}
	\bW_n&=[\bY_n-\hat{\bV}^p_n]\bmod 2^R,\nonumber
	\end{align}
	where the modulo reduction is to be understood as taken component-wise. 
	\item Define the $p$th order prediction error $\bE_n^p\triangleq \bV_n-\hat{\bV}_n^p$, and compute its covariance matrix $\bm{\Sigma}_{p}=\mathbb{E}\left[\bE_n^p (\bE_n^p)^T\right]$ based on $\{C_{\ell m}[r]\}$ for all $\ell,m\in \{1,\ldots, K\}$ and $\alpha$.
	Note that $\bm{\Sigma}_{p}$ is indeed invariant with respect to $n$ due to stationarity.
	\item Solve
	\begin{align}
	\bA&=[\ba_1|\cdots|\ba_K]^T\nonumber\\
	&=\argmin_{\substack{{\bar{\bA}\in\ZZ^{K\times K}}\\{|\bar{\bA}|\neq 0}}}\max_{k=1,\ldots,K} \frac{1}{2}\log\left(12\bar{\ba}_k^T\bm{\Sigma}_p\bar{\ba}_k\right).\label{eq:IFSCST}
	\end{align}\label{st:IntegerStep}
	\item For $k=1,\ldots,K$, compute
	\begin{align}
	\bar{g}^k_{n}&\triangleq \left[\ba_k^T\bW_n\right]\bmod 2^R\nonumber\\
	\tilde{g}^k_{n}&\triangleq\left[\bar{g}^k_{n}+\frac{1}{2}2^R\right]\bmod 2^R-\frac{1}{2}2^R,\nonumber
	\end{align}
	and set $\tilde{\bg}_n=[\tilde{g}^1_{n},\ldots,\tilde{g}^k_{n}]^T$.
	\item Compute
	\begin{align}
	\hat{\bE}^p_{n}=\bA^{-1} \tilde{\bg}_n, \ \ \hat{\bV}_{n}=\hat{\bV}_n^p+\hat{\bE}^p_{n}, \ \ \hat{\bX}_n=\frac{\hat{\bV}_n+\bf{\frac{1}{2}}}{\alpha}.\nonumber
	\end{align}
\end{enumerate}

\begin{proposition}
	Let $\bA=[\ba_1|\cdots|\ba_K]^T$ be the matrix found in step~\ref{st:IntegerStep} of the algorithm above, and define \begin{align}
	R^{\text{ST}}_{\text{IFSC}}(\bA)=\max_{k=1,\ldots,K}\frac{1}{2}\log\left(12\ba_k^T \bm{\Sigma}_p\ba_K^T\right). 
	\end{align}
	We have that
	\begin{align}
	\Pr(\mathcal{E}_{\mathrm{OL}_n})=\Pr(\hat{\bV}_n\neq \bV_n)\leq 2K\exp\left\{-\frac{3}{2}\cdot 2^{2\left(R-R^{\text{ST}}_{\text{IFSC}}(\bA)\right)} \right\},\nonumber
	\end{align}
	and
	\begin{align}
	D^k_{n}&=\mathbb{E}\left[\left(X^k_{n}-\hat{X}^k_{n}\right)^2\bigg|\OLn\right]\leq \frac{1}{12\alpha^2 (1-\Pr(\mathcal{E}_{\mathrm{OL}_n}))},\nonumber
	\end{align}
	for all $k=1,\ldots,K$, where the event $\OLn=\{\hat{\bV}_n= \bV_n\}$ is the complement of the event $\mathcal{E}_{\mathrm{OL}_n}=\{\hat{\bV}_n\neq \bV_n\}$.
	\label{prop:IFST}
\end{proposition}

\begin{proof}
 We first note that
\begin{align}
\bW_n&=[\bY_n-\hat{\bV}^p_n]\bmod 2^R\nonumber\\
&=\left[[\bV_n]\bmod 2^R-\hat{\bV}_n^p \right]\bmod 2^R\nonumber\\
&=\left[\bV_n-\hat{\bV}_n^p \right]\bmod 2^R\nonumber\\
&=\left[\bE_n^p \right]\bmod 2^R,\nonumber
\end{align}
where the second equality follows from the modulo distributive law~\eqref{eq:moddist}. By~\eqref{eq:iden}, we have that 
\begin{align}
\bar{g}^k_{n}&\triangleq \left[\ba_k^T\bW_n\right]\bmod 2^R=\left[\ba_k^T\bE^p_n\right]\bmod 2^R=[g_n^k]\bmod 2^R,\nonumber
\end{align} 
where
\begin{align}
g^k_n=\ba_K^T\bE_n^p.
\end{align}
Furthermore, $\tilde{g}^k_n\in[-\tfrac{1}{2}2^R,\tfrac{1}{2}2^R)$ is merely a cyclicly shifted version of $\bar{g}^k_n\in[0,2^R)$. Thus, $\tilde{g}^k_n=g^k_n$ if and only if $g^k_n\in[-\tfrac{1}{2}2^R,\tfrac{1}{2}2^R)$. Consequently, $\hat{\bE}^p_{n}\neq\bE_n$, and therefore $\hat{\bV}_n\neq \bV_n$, if and only if the event
\begin{align}
\mathcal{E}_{\mathrm{OL}_n}=\bigcup_{k=1}^K\left\{|g^k_{n}|\geq \frac{1}{2} 2^R \right\},\nonumber
\end{align}
occurs. Now, repeating the same steps from the proof of Proposition~\ref{prop:IFSC}, we arrive at the claimed bounds.
\end{proof}

Using Shannon's lower bound, and applying similar arguments as in~\cite{zb99}, one can show that any quantization scheme for the source $\{\bX_n\}$ that produces $R$ bits/sample/coordinate and attains $\mathbb{E}(X^k_{n}-\hat{X}^k_{n})^2\leq D$, $k=1,\ldots,K$, $n=1,\ldots$, must have $R\geq \frac{1}{K}h(\bX_n|\bX_{n-1},\ldots)-\frac{1}{2}\log(2\pi e D)$. Let $\bE_{n}^{p*}=\bX_n-\hat{\bX}_{n}^p$, where $\bX_{n}^p$ is the optimal $p$th order MMSE (linear) predictor of $\bX_n$ from $\{\bX_{n-1},\ldots,\bX_{n-p}\}$, and let $\bm{\Sigma}^*_{p}=\mathbb{E}\left[\bE_{n}^{p*}(\bE_{n}^{p*})^T\right]$. We have that
\begin{align}
h(\bX_n&|\bX_{n-1},\ldots,\bX_{n-p})=h(\bE_n^{p*}|\bX_{n-1},\ldots,\bX_{n-p})\nonumber\\
&\stackrel{(a)}{=}h(\bE_n^{p*})\stackrel{(b)}{=}\frac{1}{2}\log\left((2\pi e)^K |\bm{\Sigma}_{p}^*|\right),\nonumber
\end{align}
where $(a)$ follows from the orthogonality principle of MMSE estimation~\cite{gershogray}, and $(b)$ from the fact that $\bE_n^{p*}$ is a Gaussian random vector~\cite{coverthomas}. Thus, for any quantization scheme we must have
\begin{align}
R(D)\geq R_{\text{SLB}}(D)\triangleq \frac{1}{2}\log\left(\frac{\lim_{p\to\infty}\left|\bm{\Sigma}_{p}^*\right|^{\frac{1}{K}}}{D}\right).\nonumber
\end{align}

Similarly to previous subsections, we set $D=1/12\alpha^2$, which is a good approximation for $D^n_k$, $k=1,\ldots,K$, provided that $\delta=R-R^{\text{ST}}_{\text{IFSC}}(A)$ is not too small. The rate required by our scheme, as given in Proposition~\ref{prop:IFST}, depends on $12\bm{\Sigma}_p$, which corresponds to the prediction error covariance of the process $\tilde{\bX}_n=\sqrt{12\alpha^2}\bX_n+\tilde{\bZ}_n=\tfrac{1}{\sqrt{D}}(\bX_n+\sqrt{D}\tilde{\bZ}_n)$, where $\tilde{\bZ}_n=\sqrt{12}\bZ_n$ is a random vector with unit variance iid entries. Let $\tilde{\bm{\Sigma}}_p$ be the $p$th order prediction error covariance of the process $\bX_n+\sqrt{D}\tilde{\bZ}_n$. We can rewrite the rate required by our scheme as
\begin{align}
R^{\text{ST}}_{\text{IFSC}}(\bA,D)\triangleq \frac{1}{2}\log\left(\frac{\max_{k=1,\ldots,K}\ba_k^T\tilde{\bm{\Sigma}}_p\ba_k}{D}\right).\nonumber
\end{align}
Now, noting that if $h(\bX_n|\bX_{n-1},\ldots)>-\infty$, we have that  $\tilde{\bm{\Sigma}}_p\to\bm{\Sigma}_p^*$ as $D\to 0$, we obtain the following proposition.

\begin{proposition}
Assume $h(\bX_n|\bX_{n-1},\ldots)\geq -\infty$, and let $\bm{\Sigma}^*\triangleq\lim_{p\to\infty} \bm{\Sigma}_p^*$. We have that
\begin{align}
&\lim_{D\to 0}\lim_{p\to\infty} R^{\text{ST}}_{\text{IFSC}}(\bA,D)- R_{\text{SLB}}(D)\nonumber\\
&=\frac{1}{2}\log\left(\frac{\max_{k=1,\ldots,K}\ba_k^T\bm{\Sigma}^*\ba_k}{\left|\bm{\Sigma}^*\right|^{\frac{1}{K}}}\right).
\label{eq:gapST}
\end{align}
\end{proposition}

Thus, in the high-resolution regime, when taking large enough $p$, the gap between $R^{\text{ST}}_{\text{IFSC}}(\bA,D)$ and the information theoretic lower bound is dictated by the loss of IFSC for a source whose covariance vector is $\bm{\Sigma}^*$. The right hand side of~\eqref{eq:gapST} is non-negative~\cite{oe13b}, but is typically quite small. To illustrate this, we generate two correlated processes $\{X^1_{n}\}$ and $\{X^2_{n}\}$ as follows: let $\{W^1_n\},\{W^2_n\},\{W^3_n\}$ be three iid $\mathcal{N}(0,1)$ random processes. Let $X_n^1=\sum_{i=0}^{L-1} h_{i}W^3_{n-i}+W_n^1$, and $X_n^1=\sum_{i=0}^{L-1} g_{i}W^3_{n-i}+W_n^2$, where $\{h_n\}$ and $\{g_n\}$ are two filters, each with $L$ taps. Clearly, when the filters have sufficiently strong taps the process $\{\bX_n\}=[\{X_n^1\},\{X_n^2\}]^T$ will be highly correlated in time and in space. In Figure~\ref{fig:st} we plot the average rate required by the developed scheme, as well as $R_{\text{SLB}}(D)$, and the rate required by a standard ADC, denoted $R_{\text{naive}}(D)$, with respect to to an iid $\mathcal{N}(0,100)$ distribution on the $2L$ taps of $\{h_n\}$ and $\{g_n\}$. In the simulations performed, we took $L=5$ and $p=24$.

\begin{figure}[t]
\includegraphics[width=0.45\textwidth]{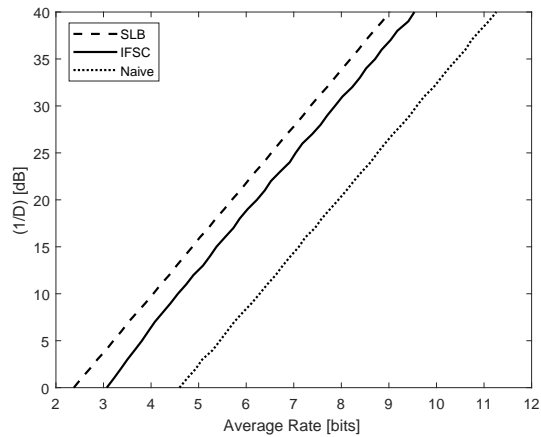}
\caption{Comparison between the average quantization rates $R^{\text{ST}}_{\text{IFSC}}(D)$, $R_{\text{SLB}}(D)$, and  $R_{\text{naive}}(D)$. The setup is that of quantizing vector of stationary processes $\{X_n^1\},\{X_n^2\}$ described in the end of Section~\ref{subsec:spacetimecor}, with $L=5$ and $p=24$.}
\label{fig:st}
\end{figure}

\section{Conclusions and Outlook}
\label{sec:conc}

We have studied the modulo ADC architecture as an alternative approach for analog-to-digital conversion. The modulo ADC allows exploitation of the statistical structure of the input process digitally at the decoder without requiring the ADC to adapt itself to the input statistics. We have demonstrated the effectiveness of oversampled modulo ADCs as a simple substitute to $\Sigma\Delta$ converters, allowing an increase in the filter's order far beyond that which is possible in current $\Sigma\Delta$ converters, since for modulo ADC filtering is done digitally. Moreover, we have shown that, when used for digitizing jointly stationary processes, parallel modulo ADCs can efficiently exploit both temporal and spatial correlations.

An implementation of modulo ADCs via ring oscillators was developed, and the corresponding input-output function for the obtained modulo ADC was characterized in terms of the delay--$V_{dd}$ profile of the inverters that construct the ring oscillator. We have then numerically studied the performance this implementation can attain for oversampled input processes, and compared it to those of $\Sigma\Delta$ converters.

There are several important challenges for future research. Perhaps most important is building a modulo ADC chip prototype. Although our simulations are based on the function $f(\cdot)$ measured from an actual (PSpice model of a) ring oscillator device, a hardware implementation is needed to fully assess the benefits of modulo ADCs. Furthermore, we would like to see whether it is possible to construct inverters with more favorable properties for ring oscillator-based modulo ADCs. In particular, we would like them to have a larger range where they are well approximated by an affine function. Another interesting avenue for future research is finding functions $g(\cdot)$ that can be implemented in the analog domain, such that the composition of function $f\circ g=f(g(\cdot))$ is more linear.

\section*{Acknowlegement}
The authors are deeply grateful to Uri Erez, whose
humbleness is the only reason for his absence from the authors list.

\bibliographystyle{IEEEtran}
\bibliography{modADCbib}

\end{document}